\documentclass[a4paper, amsfonts, amssymb, amsmath, reprint, showkeys, nofootinbib, twoside, floatfix,superscriptaddress, aps, pra]{revtex4-1}
\usepackage[english]{babel}
\usepackage{braket}

\usepackage{siunitx}
\usepackage{amsthm}
\usepackage{mathtools}
\usepackage{physics}
\usepackage{xcolor}
\usepackage{graphicx}
\usepackage[left=23mm,right=13mm,top=35mm,columnsep=15pt]{geometry} 
\usepackage{adjustbox}
\usepackage{placeins}
\usepackage[T1]{fontenc}
\usepackage{lipsum}
\usepackage{csquotes}
\usepackage{amsmath,amssymb}
\usepackage{graphicx}
\usepackage{tabularx,booktabs}
\usepackage{tikz}
\usepackage{float}
\usetikzlibrary{quantikz2}
\usepackage[utf8]{inputenc}
\usepackage[colorinlistoftodos, color=green!40, prependcaption]{todonotes}
\usepackage[pdftex, pdftitle={Article}, pdfauthor={Author}]{hyperref}
\usepackage[ruled,lined]{algorithm2e}
\usepackage{algpseudocode}
\usepackage{hyperref}

\usepackage{caption}

\newtheorem{theorem}{Theorem}

\newtheorem{lemma}{Lemma}
\newtheorem{corollary}{Corollary}

\theoremstyle{definition}

\newtheorem{defi}{Definition}

\usetikzlibrary{shapes.geometric, arrows}
\usetikzlibrary{automata,positioning}

\tikzstyle{startstop} = [rectangle, rounded corners, minimum width=3cm, minimum height=1cm, text centered, draw=black, fill=blue!20]
\tikzstyle{io} = [trapezium, trapezium left angle=70, trapezium right angle=110, minimum width=3cm, minimum height=1cm, text centered, draw=black, fill=blue!30]
\tikzstyle{process} = [rectangle, minimum width=3cm, minimum height=1cm, text centered, draw=black, fill=blue!20]
\tikzstyle{decision} = [diamond, minimum width=3cm, minimum height=1cm, text centered, draw=black, fill= green!30]
\tikzstyle{arrow} = [thick, ->, >=stealth]

\begin{document}
\title{Accelerating quantum imaginary-time evolution with random measurements}
\author{Ioannis Kolotouros}
    \email{I.Kolotouros@sms.ed.ac.uk}
    \affiliation{SandboxAQ, Palo Alto, USA}
    \affiliation{University of Edinburgh, School of Informatics, EH8 9AB Edinburgh, United Kingdom}
\author{David Joseph}
    \email{david.joseph@sandboxaq.com}
    \affiliation{SandboxAQ, Palo Alto, USA}
\author{Anand Kumar Narayanan}
    \email{anand.kumar@sandboxaq.com}
    \affiliation{SandboxAQ, Palo Alto, USA}
    
\date{\today}

\begin{abstract}
Quantum imaginary-time evolution (QITE) is a promising tool to prepare thermal or ground states of Hamiltonians, as convergence is guaranteed when the evolved state overlaps with the ground state. However, its implementation using a a hybrid quantum/classical approach, where the dynamics of the parameters of the quantum circuit are derived by McLachlan's variational principle is impractical as the number of parameters $m$ increases, since each step in the evolution takes $\Theta(m^2)$ state preparations to calculate the quantum Fisher information matrix (QFIM). In this work, we accelerate QITE by rapid estimation of the QFIM, while conserving the convergence guarantees to the extent possible. To this end, we prove that if a parameterized state is rotated by a 2-design and measured in the computational basis, then the QFIM can be inferred from partial derivative cross correlations of the probability outcomes.  One sample estimate costs only $\Theta(m)$ state preparations, leading to rapid QFIM estimation when a few samples suffice. The second family of estimators take greater liberties and replace  QFIMs with averaged classical Fisher information matrices (CFIMs). In an extreme special case optimized for rapid (over accurate) descent, just one CFIM sample is drawn. We justify the second estimator family by proving rapid descent. Guided by these results, we propose the \emph{random-measurement imaginary-time evolution} (RMITE) algorithm, which we showcase and test in several molecular systems, with the goal of preparing ground states.

\end{abstract}

\maketitle

\section{Introduction}

We are currently entering the \emph{early fault-tolerant era}, where quantum hardware improves at a fast pace, the (still small in number) qubits achieve error-correction \cite{da2024demonstration, bluvstein2024logical}, and indications of useful quantum results \cite{kim2023evidence} 
start to appear. Most of the industry and academic research is focused on finding practical applications where quantum computers can offer an advantage.

One approach is to run a quantum computer in conjunction with a classical supercomputer. The former can generate and measure non-classically-simulatable quantum states and the latter can classically process these measurements to inform further quantum computations. Examples of such frameworks include variational quantum algorithms \cite{cerezo2021variational}, variational imaginary-time evolution \cite{mcardle2019variational, gacon2024variational}, and classical shadows \cite{huang2020predicting, de2023classical, sack2022avoiding}.

A \emph{killer-app} for quantum computers is widely believed to be the ground state preparation of complex quantum mechanical systems \cite{chen2023quantum, dong2022ground, lin2020near, motlagh2024ground} such as molecules; a task in which quantum computers may offer an exponential quantum advantage, but this is still to be confirmed \cite{lee2023evaluating}. Popular candidates include quantum Gibbs samplers \cite{motta2020determining} where the physical system is coupled with a low-temperature thermal bath, or quantum imaginary-time evolution (QITE) \cite{motta2020determining, mcardle2019variational}, where under certain condition if the evolution is imaginary, the system will find itself onto its ground state.

In this work, we propose improvements to QITE, which still suffers exorbitant performance penalties. Here, the quantum state is evolved under the non-unitary operator $e^{-H\tau}$ where $\tau \in \mathbb{R}$. Such an evolution is guaranteed to converge to the ground state as long as the initial state is prepared to have a non-negligible overlap with the ground state. However, such an operator is \emph{non-physical}, meaning that in order to be realized and executed in quantum hardware, it must be transformed into a hardware-compatible quantum channel.

Real-time evolution can be realized in a quantum device using several different approaches, such as the Suzuki-trotter approximation \cite{suzuki1991general}, quantum signal processing \cite{low2017optimal, low2019hamiltonian}, or variational approaches \cite{yuan2019theory, benedetti2021hardware} to name a few. On the other hand, imaginary-time evolution requires different methods in order to realize it efficiently (and practically) on a quantum computer \cite{motta2020determining, mcardle2019variational, gacon2024variational, Fitzek2024optimizing, gacon2023stochastic}. For example, in \cite{motta2020determining}, the authors argued that one could trotterize the non-unitary evolution and, for each trotter step, find the unitary operator that is closest (with respect to 2-norm) to the former. 
In \cite{mcardle2019variational}, the authors argued that one could employ a parameterized quantum circuit and find the parameter dynamics that follow the imaginary-time evolution; a method called variational quantum imaginary-time evolution (VarQITE). Additionally, in \cite{huo2023error}, the authors used quantum Monte Carlo to simulate imaginary-time evolution and \cite{cao2022quantum} employed reinforcement learning techniques to mitigate the error induced by Trotterization and local approximation errors.

While each approach has its merits, they also come with bottlenecks. In the variational case \cite{mcardle2019variational}, in each step, the user has to calculate the quantum Fisher information matrix (QFIM) \cite{mcardle2019variational, meyer2021fisher}. This is extremely costly for applying QITE in current quantum devices. As noted in \cite{gacon2024variational}, the implementation of 200 VarQITE iterations (i.e. calculations of Eq. \eqref{eq:diff_eq_imaginary}) in state-of-the-art superconducting quantum processors using a $\approx 600$ parameter circuit would require times close to a year. The authors were able to reduce the time to a week by replacing the calculation of the QFIM with a dual problem which requires the solution of a fidelity-based optimization. 

In the following, we utilize a powerful tool that has been exploited in the quantum computing literature, and that is \emph{random measurements} \cite{elben2023randomized}. We show that the QFIM can be approximated by measuring the parameterized quantum state at random bases and propose two estimators. Previous works have utilized random measurements to calculate quantities such as the Rényi entropy \cite{elben2019statistical, elben2018renyi},  to identify mixed-state entanglement \cite{elben2020mixed} and to estimate the overlap of two quantum states \cite{elben2019statistical, elben2020cross}. On top of that, they have been utilized to calculate certain observables of a quantum state \cite{huang2020predicting}, as quantum states typically carry more information than needed to calculate the observables.\\

\emph{Our Contributions:}
\begin{itemize}
    \item We prove how the QFIM can be inferred using only $O(Km)$ quantum states, where $m$ is the number of parameters in the parameterized quantum circuit and $K$ is the number of random measurements which in practice is much smaller than $m$. We then show that the random measurement can be performed by first rotating the state by a unitary that is sampled from a 2-design and then measuring in the computational basis.

    \item We propose a second estimator to the QFIM that is constructed as the average classical Fisher information matrix when the parameterized state is measured at random (by applying a random unitary operator and then measuring in the computational basis).

    \item We propose an imaginary-time evolution algorithm based on the previous estimators that requires significantly fewer quantum state preparations than VarQITE.

    \item We test our algorithm on the task of preparing ground states of different molecular systems and show that our method brings hybrid approaches a step closer to practical implementations.
    
\end{itemize}

\emph{Structure.} In Sec. \ref{sec:preliminaries}, we give the essential background on the relevant information matrices and on imaginary-time evolution. In Sec. \ref{sec:approximatin_qfisher_random}, we present our main results, which are two estimators of the quantum Fisher information matrix. The former is based on the fidelity estimation using random measurements proposed in \cite{elben2019statistical} while the latter is constructed as an average classical Fisher information matrix. In the same section, we propose an algorithm that approximates imaginary-time evolution using a hybrid quantum/classical setting and is considerably faster than VarQITE. In Sec. \ref{sec:quantum_chemistry}, we compare our algorithm with VarQITE on the task of preparing ground states of certain molecular systems and showcase its advantage. We conclude in Sec. \ref{sec:discussion} with a general discussion of our results and future work.

\section{Preliminaries}
\label{sec:preliminaries}

In this section, we first describe standard distance measures in the space of parameterized probability distributions/quantum states and derive metrics that characterize the underlying local geometry. Then, we set the stage by outlining imaginary time evolution. 

\subsection{Quantum Fisher Information Matrix}
In this paper, we will consider pure quantum states $\ket{\phi(\boldsymbol{\theta})}$ that are parameterized by a real $m$-dimensional vector $\boldsymbol{\theta}\in \mathbb{R}^m$ through a smooth map $\boldsymbol{\theta} \mapsto \ket{\phi(\boldsymbol{\theta})}$. Distance between two such states $\ket{\phi(\boldsymbol{\theta})}, \ket{\phi(\boldsymbol{\theta'})}$ is typically measured through \emph{infidelity}, defined as
\begin{equation}
d_F\big(\ket{\phi(\boldsymbol{\theta})}, \ket{\phi(\boldsymbol{\theta'})}\big) := 1 - |\bra{\phi(\boldsymbol{\theta})}\ket{\phi(\boldsymbol{\theta'})}|^2
\label{eq:infidelity}
\end{equation}
If we allow the parameters $\boldsymbol{\theta},\boldsymbol{\theta'}$ to differ only by a small vector $\boldsymbol{\epsilon}$ (with $\norm{\boldsymbol{\epsilon}}$ being small), then the Taylor expansion of Eq. \eqref{eq:infidelity} truncated to neglect third-order terms equals
\begin{equation}
d_F(\ket{\phi(\boldsymbol{\theta})},\ket{\phi(\boldsymbol{\theta}+ \boldsymbol{\epsilon})}) = \frac{1}{4} \boldsymbol{\epsilon}^T [\mathcal{F}_Q(\boldsymbol{\theta})] \boldsymbol{\epsilon} = \frac{1}{4} \norm{\boldsymbol{\epsilon}}_{\mathcal{F}_Q}^2,
\end{equation}
where $\mathcal{F}_Q(\boldsymbol{\theta})$ is the quantum Fisher information matrix (QFIM) at $\boldsymbol{\theta}$ defined as the Hessian of the infidelity:
\begin{equation}
    \mathcal{F}_Q(\boldsymbol{\theta}) := 2\grad^2 d_F(\ket{\phi(\boldsymbol{\theta})},\ket{\phi(\boldsymbol{\theta}+ \boldsymbol{\epsilon})}) \Big{|}_{\boldsymbol{\epsilon}=0}
\end{equation}
The matrix elements can be expressed as (see \cite{stokes2020quantum, meyer2021fisher})
\begin{equation}
\begin{gathered}
[\mathcal{F}_Q(\boldsymbol{\theta})]_{ij} = 4 \; \mathrm{Re}\Bigg[\frac{ \partial \bra{\phi(\boldsymbol{\theta})}}{\partial \theta_i} \frac{\partial \ket{\phi(\boldsymbol{\theta})}}{\partial \theta_j} \\
   - \frac{ \partial \bra{\phi(\boldsymbol{\theta})}}{\partial \theta_i}\ket{\phi(\boldsymbol{\theta})}\bra{\phi(\boldsymbol{\theta})}\frac{ \partial \ket{\phi(\boldsymbol{\theta})}}{\partial \theta_j} \Bigg].
\label{eq:qfim_elements}
\end{gathered}
\end{equation}

\subsection{Classical Fisher Information Matrix}
Just as we defined distances in the space of quantum states, we can define distances in the space of probability distributions. However, in our quantum setting, the probability distributions depend on the choice of measurement basis $\mathcal{M}$. In this paper, we will focus on measurements that can be performed by first applying a global unitary $U$ on the quantum state and then measuring in the computational basis. The probability $p_{\boldsymbol{s}}^U$ of each outcome $\boldsymbol{s}\in \{0, 1\}^n$ is given as:
\begin{equation}
    p_{\boldsymbol{s}}^U = \tr ( U \rho U^\dagger \Pi_{\boldsymbol{s}})
\end{equation}
where $\Pi_{\boldsymbol{s}} = \ketbra{\boldsymbol{s}}$ is the projection operator on the $\boldsymbol{s}$-th eigenspace.

Consider the probability distributions $\boldsymbol{p}_{U}(\boldsymbol{\theta})$ and $ \boldsymbol{p}_{U}(\boldsymbol{\theta}+ \boldsymbol{\epsilon})$ resulting after rotating the states $\ket{\phi(\boldsymbol{\theta})}$ and $\ket{\phi(\boldsymbol{\theta+\boldsymbol{\epsilon}})}$ by a unitary $U$ and then measuring in the computational basis (see Figure \ref{fig:general_measurement_circuit}). Let the distance measure be the (\emph{Kullback-Leibler}) KL-divergence (or else the relative entropy) defined as
\begin{equation}
    \mathrm{KL}(\boldsymbol{u}||\boldsymbol{v}) := \sum_{j=1}^K u_j \log \frac{u_j}{v_j}
\end{equation}
for any $\boldsymbol{u}, \boldsymbol{v}\in \Delta^{K-1}$ with $\Delta^{K-1}$ being the probability simplex of dimension $K-1$. If the shift vector $\boldsymbol{\epsilon}$ is small, then the KL-divergence can be expressed as (if we again neglect third-order terms)
\begin{equation}
 \mathrm{KL}(\boldsymbol{p}_U(\boldsymbol{\theta})|| \boldsymbol{p}_U(\boldsymbol{\theta+\boldsymbol{\epsilon}})) = \frac{1}{2}\boldsymbol{\epsilon}^T [\mathcal{F}_C^U(\boldsymbol{\theta})] \boldsymbol{\epsilon} = \frac{1}{2}\norm{\boldsymbol{\epsilon}}_{\mathcal{F}_C^{U}}^2,
\end{equation}
where $\mathcal{F}_C^U$ is the classical Fisher information matrix (CFIM) \cite{meyer2021fisher, kolotouros2024random} whose elements are defined as
\begin{equation}
    [\mathcal{F}_C^U(\boldsymbol{\theta})]_{ij} := \sum_{\boldsymbol{s}} \frac{1}{p_{\boldsymbol{s}}^U(\boldsymbol{\theta})}\frac{\partial p_{\boldsymbol{s}}^U(\boldsymbol{\theta})}{\partial \theta_i} \frac{\partial p_{\boldsymbol{s}}^U(\boldsymbol{\theta})}{\partial \theta_j}.
\end{equation}
where $p_{\boldsymbol{s}}^U(\boldsymbol{\theta}) = \Tr [U\rho(\boldsymbol{\theta})U^\dagger \Pi_{\boldsymbol{s}}]$ and $\rho(\boldsymbol{\theta}) = \ketbra{\phi(\boldsymbol{\theta})}$. The partial derivatives $\frac{\partial p_{\boldsymbol{s}}^U(\boldsymbol{\theta})}{\partial \theta_i}$ can be estimated using the parameter-shift rules \cite{mari2020estimating, schuld2019evaluating, wierichs2022general} as
\begin{equation}
    \frac{\partial p_{\boldsymbol{s}}^U(\boldsymbol{\theta}) }{\partial \theta_j} = \frac{1}{2}\Big( p_{\boldsymbol{s}}^U\Big(\boldsymbol{\theta} + \frac{\pi}{2}\hat{\boldsymbol{e}}_j\Big) - p_{\boldsymbol{s}}^U\Big( \boldsymbol{\theta} - \frac{\pi}{2}\hat{\boldsymbol{e}}_j\Big)\Big).
\end{equation}

\begin{figure}
    \centering
    \includegraphics[scale=0.5]{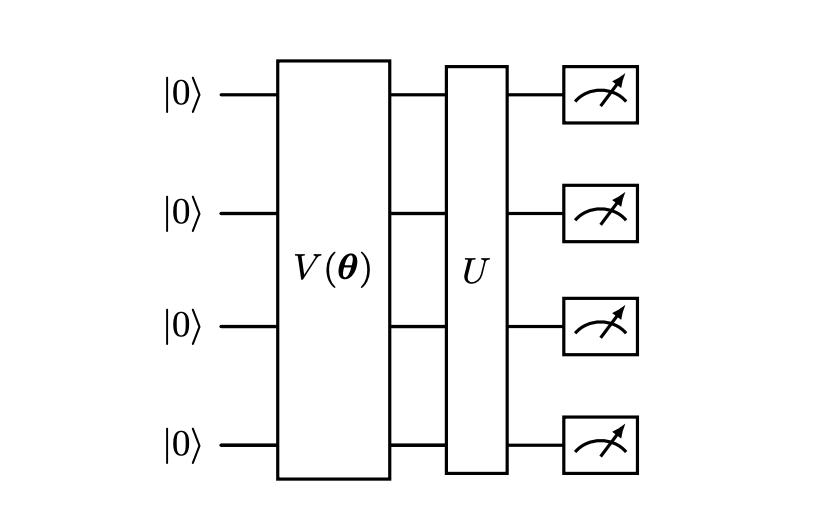}
    \caption{General random-measurement framework rotating the parameterized quantum state by a unitary $U \sim \nu$ where $\nu \subseteq U(2^n)$ and then measuring in the computational basis.}
\label{fig:general_measurement_circuit}
\end{figure}

\subsection{Quantum imaginary-time evolution}
\label{subsec:imaginary_time_evolution}

In quantum mechanics, when a system is initialized in the quantum state $\ket{\psi(0)}$ (at time $t=0$) and its dynamics are described by a time-independent Hamiltonian $H$, it will evolve under the unitary $e^{-iHt}$, i.e.:
\begin{equation}
    \ket{\psi(t)} = e^{-iHt}\ket{\psi(0)}
 \end{equation}
Such an evolution can be simulated in a gate-based quantum computer by ``trotterizing" the unitary evolution $e^{-iHt}$ into short time intervals $\delta t$. 

If we allow the time to take imaginary values ($\tau \equiv it$), then the operator $e^{-H\tau}$ is no longer unitary, and the evolution is called \emph{imaginary-time evolution}. As a first step, we will derive the mathematical equation that governs the imaginary-time evolution. Consider the imaginary-time evolved state $\ket{\psi(\tau)}$:
\begin{equation}
    \ket{\psi(\tau)} = A(\tau)e^{-H\tau}\ket{\psi(0)}
\end{equation}
where:
\begin{equation}
    A(\tau) = \Bigg(\frac{1}{\sqrt{\bra{\psi(0)}e^{-2H\tau}\ket{\psi(0)}}}\Bigg)
\end{equation}
is a normalization factor that ensures that the imaginary-evolved quantum state is normalized, i.e. $\bra{\psi(\tau)}\ket{\psi(\tau)}=1$. The evolution under the imaginary-time evolution is governed by the \emph{Wick-Schr\"{o}dinger} equation. To see this, we take the time derivative:
\begin{equation*}
\begin{aligned}
    \frac{\partial \ket{\psi(\tau)}}{\partial \tau} = \frac{\partial}{\partial \tau} \Big(A(\tau)e^{-H\tau}\ket{\psi(0)}\Big) =\\
    \frac{\partial A(\tau)}{\partial \tau}e^{-H\tau}\ket{\psi(0)} + A(\tau)\frac{\partial e^{-H\tau}}{\partial \tau} \ket{\psi(0)}
\end{aligned}
\end{equation*}
Computing the derivative in the first term, we obtain:
\begin{equation}
    \frac{\partial A(\tau)}{\partial \tau} = \frac{\partial}{\partial \tau}\Bigg(\frac{1}{\sqrt{\bra{\psi(0)}e^{-2H\tau}\ket{\psi(0)}}}\Bigg) = A(\tau)E_\tau
\end{equation}
where $E_\tau = \bra{\psi(\tau)}H\ket{\psi(\tau)}$.
Thus, putting everything back together we obtain the \emph{Wick-Schr\"{o}dinger equation}:
\begin{equation}
\begin{aligned}
    \frac{\partial \ket{\psi(\tau)}}{\partial \tau} = 
    (E_\tau - H)\ket{\psi(\tau)}
\end{aligned}
\end{equation}

As we discussed, imaginary-time evolution is a very interesting tool that allows the preparation of thermal states \cite{motta2020determining, PhysRevA.108.022612,turro2023quantum} or ground states \cite{mcardle2019variational, gomes2021adaptive}. The necessary condition is that the initial state is prepared with a non-zero overlap with the ground state of the Hamiltonian of interest. To see this, consider a Hamiltonian $H$ and an initial state $\ket{\psi(0)}$ that has a non-zero overlap with the ground state $\ket{\psi_0}$. We can write the initial state in the energy eigenbasis as:
\begin{equation}
    \ket{\psi(0)} = a_0 \ket{\psi_0} + \sum_{j\neq 0}a_j \ket{\psi_j} 
\end{equation}
where $\ket{\psi_0}$ is the ground state. Evolving the state according to the imaginary time-evolution will result in the quantum state:
\begin{equation}
\begin{gathered}
    \ket{\psi(\tau)} = A(\tau) \Big[a_0 e^{-H\tau}\ket{\psi_0} + \sum_{j\neq 0}a_j e^{-H\tau}\ket{\psi_j}\Big] \\
    = A(\tau) \Big[a_0 e^{-E_0\tau}\ket{\psi_0} + \sum_{j\neq 0}a_j e^{-E_j \tau}\ket{\psi_j}\Big]
\end{gathered}
\end{equation}
As a result, in the limit of $\tau \rightarrow \infty$ the system reaches the ground state.

In our case, we are equipped with a small-scale (and perhaps noisy) quantum computer and we aim to approximate the exact imaginary-time evolution described by the states $\ket{\psi(\tau)}$ by a family of parameterized state $\ket{\phi(\boldsymbol{\theta}(\tau)}$ \cite{mcardle2019variational} which approximate the former states as much as possible. In other words, we aim to find the parameter dynamics $\boldsymbol{\theta}(\tau)$ so that the parameterized state approximates the imaginary-time evolution. Starting from McLachlan's variational principle \cite{mclachlan1964variational}:
\begin{equation}
    \delta \norm{(d/d\tau + H - E_\tau)\ket{\phi(\boldsymbol{\theta}(\tau))}}_2 = 0 
\end{equation}
and introducing a time-dependent global phase in the calculation \cite{yuan2019theory}, we find (see \cite{mcardle2019variational, yuan2019theory} for details) that the parameters must satisfy:
\begin{equation}
\mathcal{F}_Q(\boldsymbol{\theta}(\tau)) \dot{\boldsymbol{\theta}} = -2\grad_{\boldsymbol{\theta}} E_\tau(\boldsymbol{\theta}(\tau))
\label{eq:diff_eq_imaginary}
\end{equation}
where $\mathcal{F}_Q$ is the quantum Fisher information matrix defined in Eq. \eqref{eq:qfim_elements}. One of the major drawbacks is that the evaluation of Eq. \eqref{eq:diff_eq_imaginary} at a certain point requires the preparation of $\Theta(m^2)$ quantum states, scaling quadratically with the number of parameters. On top of that, in order to calculate each element of the QFIM, one has to prepare circuits of size twice the depth of those needed to prepare the quantum state $\ket{\phi(\boldsymbol{\theta})}$ \cite{mari2020estimating}, or otherwise to employ a Hadamard-overlap test with twice the qubit resources \cite{haug2021capacity}.

\section{Approximating quantum Fisher information matrix with random measurements}
\label{sec:approximatin_qfisher_random}

As previously discussed in Subsection \ref{subsec:imaginary_time_evolution}, there is a need for a fast calculation of the quantum Fisher information matrix (QFIM) or an approximation to it. In this section, we outline our results on the approximation of the QFIM using random measurements \cite{elben2023randomized} and defer the proofs to the appendix. Throughout the rest of the manuscript we will denote $\mathcal{F}_Q$ the quantum Fisher information matrix and $\tilde{\mathcal{F}}_Q$ any approximation to it. Our first result is presented in Theorem \ref{th:main_theorem}. 

el
\begin{theorem}
    For every parameterization  $\boldsymbol{\theta} \mapsto \ket{\phi(\boldsymbol{\theta})}$, the matrix elements of the quantum Fisher information matrix can be inferred as 
    \begin{equation*}
    [\mathcal{F}_Q(\boldsymbol{\theta})]_{ij} = 2(2^n + 1)  \sum_{\boldsymbol{s}} \mathbb{E}_{U\sim \mu_H}\Bigg[\frac{\partial p_{\boldsymbol{s}}^U(\boldsymbol{\theta})}{\partial \theta_i}\frac{\partial p_{\boldsymbol{s}}^U(\boldsymbol{\theta})}{\partial \theta_j}\Bigg],
    \label{eq:quantum_from_random_measurements}
    \end{equation*}
where $\mathbb{E}_{U\sim \mu_H}[\cdot]$ is the ensemble average over random unitary $U$ drawn from the Haar distribution $\mu_H$ and
$p_{\boldsymbol{s}}^U(\boldsymbol{\theta}) := \bra{\phi(\boldsymbol{\theta})}U^{\dagger}\Pi_{\boldsymbol{s}}U\ket{\phi(\boldsymbol{\theta})}$ is the probability of the outcome $\boldsymbol{s}$ when measuring $U\ket{\phi(\boldsymbol{\theta})}$ with respect to the computational basis projectors $\{\Pi_{\boldsymbol{s}}\}$. 
\label{th:main_theorem}
\end{theorem}

\begin{proof}
    For a detailed proof, see Appendix \ref{appendix:random_measurements}.
\end{proof}

Each sample requires at most $2m$ quantum state preparations in total. To sample, it suffices to compute the partial derivatives of the probability outcomes in Eq. \eqref{eq:quantum_from_random_measurements}, and those can be easily calculated using parameter-shift rules (see Preliminaries section). The immediate result is that, in practice, this estimator requires significantly fewer quantum states to approximate the QFIM since it can be written as a product of first-order derivatives.

In general, generating Haar random unitaries on a quantum computer is a computationally exhaustive task since most unitary operators require a number of gates that scale exponentially to the number of qubits \cite{mele2024introduction}. On the other hand, $k$-designs are distributions that match the Haar moments up to the $k$-th order (see Definition \ref{defi:k-design}). The advantage is that $k$-designs can be generated efficiently.

\begin{defi}
    \emph{(Unitary $k$-design) 
    A probability distribution $\nu$ supported over a set of unitaries $S\subseteq U(d)$ is defined to be a unitary $k$-design if and only if}
    \begin{equation}
       \mathbb{E}_{V\sim \nu}[V^{\otimes k}O V^{\dagger \otimes k}] =  \mathbb{E}_{U\sim \mu_H} [U^{\otimes k}O U^{\dagger \otimes k}]   
    \end{equation}
    \emph{for all $O\in \mathcal{L}((\mathbb{C}^d)^{\otimes k})$}.
    \label{defi:k-design}
\end{defi}
We next prove a corollary of Theorem \ref{th:main_theorem}, recasting Haar random unitaries with $2$-designs.

\begin{corollary}\label{cor:2design}
    For $U$ drawn from a 2-design $\nu$, the elements of the quantum Fisher information matrix satisfy
      \begin{equation}
        [\mathcal{F}_Q(\boldsymbol{\theta})]_{ij} = 2(2^n + 1)    \sum_{\boldsymbol{s}} \mathbb{E}_{U\sim \nu}\Bigg[\frac{\partial p_{\boldsymbol{s}}^U(\boldsymbol{\theta})}{\partial \theta_i}\frac{\partial p_{\boldsymbol{s}}^U(\boldsymbol{\theta})}{\partial \theta_j}\Bigg]
        \label{eq:quantum_from_random_measurements_2_design}
    \end{equation} 
    where 
$p_{\boldsymbol{s}}^U(\boldsymbol{\theta}) := \bra{\phi(\boldsymbol{\theta})}U^{\dagger}\Pi_{\boldsymbol{s}}U\ket{\phi(\boldsymbol{\theta})}$ is the probability of the outcome $\boldsymbol{s}$ when measuring $U\ket{\phi(\boldsymbol{\theta})}$ with respect to the computational basis projectors $\{\Pi_{\boldsymbol{s}}\}_{\boldsymbol{s}}$. 
\end{corollary}
\begin{proof}
The expectation $\mathbb{E}_{U\sim \mu_H}\Bigg[\frac{\partial p_{\boldsymbol{s}}^U(\boldsymbol{\theta})}{\partial \theta_i}\frac{\partial p_{\boldsymbol{s}}^U(\boldsymbol{\theta})}{\partial \theta_j}\Bigg]$  
in the statement of Theorem \ref{th:main_theorem} expands as
\begin{equation}
\begin{gathered}
\Tr \Bigg[\mathbb{E}_{U\sim \mu_H}[U^{\otimes 2} \Pi_{\boldsymbol{s}}^{\otimes 2} U^{\dagger \otimes 2}] \frac{\partial \rho(\boldsymbol{\theta})}{\partial \theta_i} \otimes \frac{\partial \rho(\boldsymbol{\theta})}{\partial \theta_j}\Bigg],
    \end{gathered}
    \end{equation}
 which, by specializing Definition \ref{defi:k-design} to $2$-designs equals 
 \begin{equation*}
 \mathbb{E}_{U\sim \nu}\Bigg[\frac{\partial p_{\boldsymbol{s}}^U(\boldsymbol{\theta})}{\partial \theta_i}\frac{\partial p_{\boldsymbol{s}}^U(\boldsymbol{\theta})}{\partial \theta_j}\Bigg],    
 \end{equation*}
 thereby proving the corollary.
\end{proof}

\paragraph*{The $2$-design estimator:} Corollary \ref{cor:2design} suggests natural inference procedures for estimating the QFIM by sampling unitaries that come from a $k$-design with $k\geq 2$ (since a $k$-design is also a 2-design if $k\geq 2$). We next derive a simple estimator in this motif and call it the \emph{the 2-design estimator}. To efficiently sample the unitary, we draw from the ensemble of the $n$-qubit Clifford group $Cl(n)$, which forms a 3-design. Elements from the $n$-qubit Clifford group can be generated by a circuit with at most $\mathcal{O}(n^2/\log n)$ elementary gates \cite{aaronson2004improved}, showcasing that our 2-design estimator can be implemented in a NISQ/early fault-tolerant quantum computer. As before, computing the full gradient vector (and hence, drawing one sample of the estimate in Corollary \ref{cor:2design}) only takes $O(m)$ quantum state preparations. To obtain the estimate in Corollary \ref{cor:2design}, we repeat this sampling procedure $K$ times, where $K$ is a hyper-parameter considered as a design choice. The state preparation cost $\Theta(Km)$ of the 2-design estimator significantly improves on the $\Theta(m^2)$ required by previous known algorithms for QFIM when $K$ is small (see Appendix \ref{appendix:quantum_resources} for more details on the actual quantum resources needed). 

For most parameterizations $\boldsymbol{\theta}\mapsto \ket{\phi(\boldsymbol{\theta})}$, the measurement probability $p_{\boldsymbol{s}}^U(\boldsymbol{\theta})$ distribution is likely concentrated around its expectation. Therefore, a small $K$ likely suffices for most applications. As a curiosity to determine the accuracy limits of the 2-design estimator in the worst case, it would be interesting to construct parametrizations requiring a large $K$, perhaps by planting pathological high dimensional singularities. The hyperparameter can also be adaptively tuned, keeping it small at the early stages for rapid descent and increasing it closer towards the end of the evolution to converge to a more accurate ground state. We leave these interesting questions open for future research. We do demonstrate empirically in small examples that, in practice, $K$ is much smaller than $m$ (see Section \ref{sec:quantum_chemistry}). In essence, $K$ offers a tradeoff between rapid and accurate descent of QITE incorporating the 2-design estimator. In most cases, a small choice of $K$ is sufficiently accurate while facilitating rapid descent.  

Next, we provide a second estimator to the QFIM, which we name \emph{the average classical Fisher information matrix estimator}. As such, we provide a second definition, which is the average (over an ensemble $\nu \subseteq U(2^n)$) classical Fisher information matrix.

\begin{defi}
\emph{(Average classical Fisher information matrix). Consider an ensemble of unitary operators $\nu \subseteq U(2^n)$ from which we uniformly sample. We can define the average (over the unitary ensemble $\nu$) classical Fisher information matrix as:}
\begin{equation}
        \mathbb{E}_{U \sim \nu} [\mathcal{F}_C^U]
\label{eq:average_classical_fisher_matrix}
\end{equation}
\end{defi}

The idea is that one can get information about the underlying quantum states by measuring on a specific basis. If this procedure is performed repeatedly, one can get an accurate picture of the geometry of the parameterized quantum states. We were able to make an important observation that is true for all parameterized quantum states that were investigated in this paper.\\

\noindent \textbf{Conjecture.} \emph{If the unitaries are drawn from the Haar-distribution $\nu = \mu_H$, then the average Fisher defined in Eq. \eqref{eq:average_classical_fisher_matrix} approximates the quantum Fisher information matrix, i.e.}
\begin{equation}
    \mathbb{E}_{U \sim \mu_H} [\mathcal{F}_C^U(\boldsymbol{\theta})]  = \frac{1}{2}\mathcal{F}_Q(\boldsymbol{\theta})
\label{eq:cfisher_haar}
\end{equation}
\emph{for any $\boldsymbol{\theta}$.}\\

The proof of the previous conjecture is a very challenging task. The reason is that the unitaries that appear in the Haar-integral on the left-hand side of Eq. \eqref{eq:cfisher_haar} enter in a non-linear fashion. As such, certain results from random matrix theory cannot be directly applied in this scenario. We discuss that thoroughly in Appendix \ref{appendix:random_measurements}. The above conjecture indicates that by choosing the appropriate unitary ensemble to sample from (in this case, the Haar distribution), one can get an accurate description of the underlying geometry in the space of parameterized quantum states. However, as we later show (see Lemma \ref{lemma:descent_direction}), replacing the quantum Fisher information matrix in Eq. \eqref{eq:diff_eq_imaginary} by the average classical Fisher information matrix for any subset $\nu \subseteq U(2^n)$ will always result in a descent direction for a sufficiently small time step.

Each one of the two estimators that we propose (in Eq. \eqref{eq:quantum_from_random_measurements} and Eq. \eqref{eq:cfisher_haar}) have different advantages compared to the other. As we verified, the random CFIM estimator in Eq. \eqref{eq:cfisher_haar} outperforms the former in Eq. \eqref{eq:quantum_from_random_measurements} in terms of speed of convergence. This means that with fewer sampled unitaries (and measurements) we can approximate the quantum Fisher information to great accuracy. An illustrative example is given in Figure \ref{fig:distance from QFIM}. In this figure, we visualize how the quantum Fisher information matrix can be approximated using random measurements with either of the two estimators in Eq. \eqref{eq:quantum_from_random_measurements} and Eq. \eqref{eq:cfisher_haar} for a random 8-qubit parameterized quantum state at a random configuration $\boldsymbol{\theta}$. In this experiment, all quantities were calculated exactly by using a statevector simulator. For the average CFIM, the unitaries were drawn from the Haar distribution.

Given the same number of samples, we can see that the average CFIM estimator (blue line) is able to approximate with a smaller error the QFIM, compared to the estimator in Eq. \eqref{eq:quantum_from_random_measurements}. However, both estimators are able to achieve a good approximation to QFIM with only a small number of samples. Similar performance was observed for all ansatz families used in this manuscript. For a more detailed comparison of the actual quantum resources, see Appendix \ref{appendix:quantum_resources}.

On the other hand, the 2-design estimator in Eq. \eqref{eq:quantum_from_random_measurements} comes with other benefits. Specifically, its implementation can be performed by sampling unitaries from a 2-design, which implies that the unitaries have an exponentially smaller depth than that of the unitaries that are sampled from the Haar distribution. As such, the estimator in Eq. \eqref{eq:quantum_from_random_measurements} can even be experimentally realized in the early fault-tolerant era where the number of qubits remains small, but we can execute longer circuits. However, as we later show, the average CFIM can have a very promising performance when the unitaries are drawn from a more practical (in terms of depth required) ensemble.

Moreover, as we prove in Lemma \ref{lemma:descent_direction}, if we replace the QFIM with any average CFIM estimator (with $\nu \subseteq U(2^n)$), then for a sufficiently small time step, we will still move into a descent direction and as we show next in our experiments, this proves to be a powerful tool. The reason is that, similar to the Random Natural Gradient approach in \cite{kolotouros2024random}, one can choose a hardware-efficient ansatz as the random unitary (where the parameters are drawn at random) and then construct the random classical Fisher information matrix for that unitary. In our case, we will show that post-processing, the average of these CFIMs still approximates the QFIM with great accuracy.

\begin{figure}
\begin{tikzpicture}
\node (img)  {\includegraphics[scale=0.6]{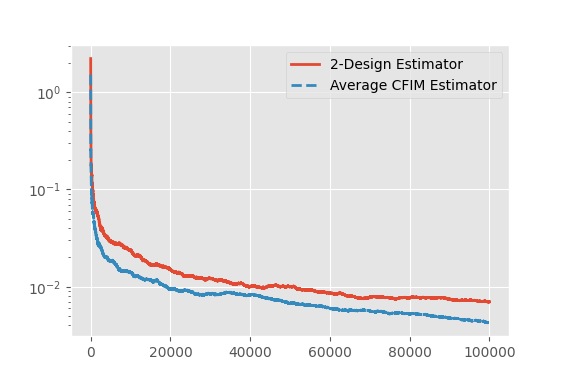}};
\node[below=of img, node distance=0cm, yshift=1.3cm] (xlabel1) {\scriptsize Samples};
\node[left=of img, node distance=0cm, rotate=90, yshift=-1.3cm, xshift=0.6cm] {\scriptsize $\norm{\tilde{\mathcal{F}}_Q - \mathcal{F}_Q}$};
\end{tikzpicture}
\caption{Distance of the quantum Fisher information from its corresponding estimators (in logarithmic scale). The red line corresponds to the estimator in Eq. \eqref{eq:quantum_from_random_measurements} while the blue line to the estimator in Eq. \eqref{eq:cfisher_haar}.}
\label{fig:distance from QFIM}
\end{figure}

The previous analyses allow us to propose a quantum algorithm that we name \emph{random-measurement imaginary-time evolution} (RMITE). The pseudoalgorithm for our method is outlined in Algorithm \ref{alg:random_measurment-ite}. The idea is that one can replace the QFIM with either of the two proposed estimators (where both require measuring the state on a random basis). In that case, the QFIM $\mathcal{F}_Q$ is replaced by its estimator $\tilde{\mathcal{F}}_Q$ and the partial differential equation is transformed to:
\begin{equation}
    \tilde{\mathcal{F}}_Q[\boldsymbol{\theta}(\tau)] \dot{\boldsymbol{\theta}} = -2\grad_{\boldsymbol{\theta}} E_\tau(\boldsymbol{\theta})
\label{eq:rmite_update}
\end{equation}

Specifically, the parameterized quantum state is rotated by a global unitary $U$, sampled by the appropriate ensemble $\nu \subseteq U(2^n)$ ($U \sim \nu$), and then is measured on the computational basis. Finally, the QFIM in Eq. \eqref{eq:diff_eq_imaginary} can be estimated by post-processing the measurement and using any of the proposed estimators.

For example, in the case where we use the average CFIM estimator $\tilde{\mathcal{F}}_Q$, the partial differential equation in imaginary-time evolution Eq. \eqref{eq:rmite_update} can be replaced by:
\begin{equation}
    \mathbb{E}_{U \sim \nu}[\mathcal{F}_C^U(\boldsymbol{\theta}(\tau)] \dot{\boldsymbol{\theta}} = -2\grad_{\boldsymbol{\theta}} E_\tau(\boldsymbol{\theta})
\label{eq:multi_unitary_approximate_varqite}
\end{equation}
It is important to stress that in the case where only a single unitary is used in Eq. \eqref{eq:cfisher_haar}, then the solution of the partial differential equation:
\begin{equation}
    \mathcal{F}_C^U (\boldsymbol{\theta}(\tau))\dot{\boldsymbol{\theta}} = -2\grad_{\boldsymbol{\theta}} E_\tau(\boldsymbol{\theta})
\label{eq:single_unitary_approximate_varqite}
\end{equation}
is equivalent to the \emph{Random Natural Gradient} \cite{kolotouros2024random}.
The error in the updated $\boldsymbol{\theta}$ is characterized by the distance between the operators, i.e. the estimator of the QFIM ($\tilde{\mathcal{F}}_Q$) and the QFIM ($\mathcal{F}_Q$). As stated in Lemma \ref{lemma:descent_direction}, the resulting update will always decrease the energy of the system. \footnote{We refer to the direction that minimizes the loss as the \emph{descent direction}.}

\begin{lemma}
    Updating the parameters of a parameterized quantum circuit according to Eq. \eqref{eq:multi_unitary_approximate_varqite} will result in a descent direction.
\label{lemma:descent_direction}
\end{lemma}

\begin{proof}
    For a detailed proof, see Appendix \ref{appendix:proof_of_lemma1}.
\end{proof}

Since Eq. \eqref{eq:single_unitary_approximate_varqite} is a special case of Eq. \eqref{eq:multi_unitary_approximate_varqite}, this implies that updating with a single CFIM will also result in a descent direction and as we later show, many times is sufficient to approximate imaginary-time evolution.  In order to apply the random measurement on the parameterized quantum state $\ket{\phi(\boldsymbol{\theta})}$, we first apply a unitary operator $U$ that is sampled from an appropriate ensemble and then the state is measured on the computational basis. By doing so, we can calculate the outcome probabilities $p_{\boldsymbol{s}}^U(\boldsymbol{\theta}) = \Tr [U\rho(\boldsymbol{\theta})U^\dagger \Pi_{\boldsymbol{s}}]$ (and consequently their first order derivatives as discussed in Preliminaries section) for projector operators $\Pi_{\boldsymbol{s}}$ with $\boldsymbol{s} \in \{0, 1\}^n$. In our experiments, we use two types of unitaries. First, for the $2$-design estimator in Eq. \eqref{eq:quantum_from_random_measurements_2_design}, we used random Clifford unitaries. Then, for the average CFIM estimator in Eq. \eqref{eq:cfisher_haar}, we choose random unitaries to be hardware-efficient parameterized circuits \cite{kandala2017hardware} with parameters sampled at random. 

Next, in the following Lemma (see Lemma \ref{lemma:error_in_approximation}), we quantify the relative error in imaginary-time evolution when we replace the QFIM with the suggested estimators.

\begin{lemma}
    Assuming that both the quantum Fisher information matrices and its estimator are full-rank and that $\dot{\boldsymbol{\theta}}_Q$ and $\dot{\tilde{\boldsymbol{\theta}}}_Q$ are given by Eqs. \eqref{eq:diff_eq_imaginary}, \eqref{eq:rmite_update}, then the relative error $\frac{\norm{\dot{\boldsymbol{\theta}}_Q - \dot{\tilde{\boldsymbol{\theta}}}_Q}}{\norm{\dot{\boldsymbol{\theta}}_Q}}$ can be upper bounded as:
    \begin{equation}
        \frac{\norm{\dot{\boldsymbol{\theta}}_Q - \dot{\tilde{\boldsymbol{\theta}}}_Q}}{\norm{\dot{\boldsymbol{\theta}}_Q}} \leq \frac{\lambda_{\max}(\mathcal{F}_Q)}{\lambda_{\min}(\tilde{\mathcal{F}}_Q)} - 1
    \end{equation}
\label{lemma:error_in_approximation}
\end{lemma}
\begin{proof}
    For a detailed proof, see Appendix \ref{appendix:error_in_approximation}.
\end{proof}

\begin{algorithm}[h!]
\caption{Random Measurement Imaginary-Time Evolution}
\label{alg:random_measurment-ite}
\SetKwInOut{Input}{Input}
\Input{Problem Hamiltonian $\mathcal{H}$\;
Initial state $\ket{\phi(\boldsymbol{\theta_0})}=U(\boldsymbol{\theta_0})\ket{\phi}$\;
Current time $t=0$\;
Total time $T$\;
Timestep $\delta t$\;
Ensemble of unitary operators $\nu \subseteq U(d)$\;}
\While{$t<T$}{
Calculate  $\grad_{\boldsymbol{\theta}} E_\tau(\boldsymbol{\theta})$\;
Estimate quantum Fisher information matrix $\tilde{\mathcal{F}}_Q$ using either Eq. \eqref{eq:quantum_from_random_measurements} or Eq. \eqref{eq:average_classical_fisher_matrix}\;
Solve $\tilde{\mathcal{F}}_Q  \dot{\boldsymbol{\theta}} = -2\grad_{\boldsymbol{\theta}} E_\tau(\boldsymbol{\theta})$\;
Update $\boldsymbol{\theta}$ as $\boldsymbol{\theta} = \boldsymbol{\theta} + \delta t \dot{\boldsymbol{\theta}}$\;
Update $t$ as $t = t + \delta t$
}
\Return $\boldsymbol{\theta}$
\end{algorithm} 

We also like to notice how one could approximate quantities such as $\mathbb{E}_{U \sim \nu} [\mathcal{F}_C^U]$ in practice. In a real-world setting, the user would select an ensemble of unitaries $\{U_i\}$ from which they would uniformly sample from. Then, they would select the number of unitaries $K$ per iteration to calculate the average. Finally, the average can be approximated as:
\begin{equation}
    \mathbb{E}_{U \sim \nu} [\mathcal{F}_C^U] \approx \frac{1}{K} \sum_{j=1}^K \mathcal{F}_C^{U_j}
\end{equation}

Another type of approximation that one could apply is based on (random) coordinate methods \cite{ding2023random, kolotouros2024random} (see \emph{stochastic-coordinate quantum natural gradient} in \cite{kolotouros2024random}). Specifically, one can ``freeze" a subset of the total parameters and optimize the rest. For example, one can calculate the reduced classical or quantum Fisher information matrix for this subset, exploiting the codependency of the parameters on this subset. The reduced information matrices still remain positive semidefinite, and thus, updating the parameters with the descent direction preconditioned by these matrices will still result in lower energy states. However, we do not investigate these approximations in this work.

\section{Molecular Ground State Preparation}
\label{sec:quantum_chemistry}

In this section, we will investigate how our proposed method performs when the task is to prepare the ground state of certain molecular Hamiltonians. Note, however, that our method is quite versatile; any mathematical problem whose solution can be mapped to the ground state of a quantum spin-interacting Hamiltonian can be tackled using RMITE. In this paper, we choose to prepare ground states of molecular systems due to the large scientific interest in this problem. We will examine how RMITE performs when different options are selected for the estimators. For example, we will examine the 2-design estimator proposed in Eq. \eqref{eq:quantum_from_random_measurements_2_design} but also the average CFIM estimators in Eq. \eqref{eq:average_classical_fisher_matrix}, when either one ($K=1$) or many CFIMs are used.

The preparation of the ground and first excited states has gained significant interest from the quantum algorithms community \cite{bauer2020quantum}. The question of whether exponential advantages can be achieved in this task remains an open question \cite{lee2023evaluating}. One approach is to use adiabatic state preparation \cite{lee2023evaluating, aspuru2005simulated, veis2014adiabatic}, where the user utilizes the adiabatic theorem that states that a system will remain at the instantaneous ground state as long as the evolution is sufficiently slow \cite{jansen2007bounds, elgart2012note}. Other approaches include quantum imaginary-time evolution \cite{kosugi2022imaginary, kosugi2023first, tsuchimochi2023improved, PhysRevA.109.052414} or hybrid quantum/classical variational approaches \cite{kolotouros2024simulating, kandala2017hardware}. Below, we give the necessary background on quantum chemistry.

Consider a molecular system with $N$ electrons in positions $\boldsymbol{r_i}$ and $M$ nuclei in positions $\boldsymbol{R_I}$ with charge $Z_I$. The molecular Hamiltonian (using atomic units) can be expressed as:

\begin{equation}
\begin{aligned}
    H = -\sum_{i=1}^N \frac{\grad_i^2}{2} - \sum_{I=1}^M \frac{\grad_I^2}{2M_I} - \sum_{i=1}^N\sum_{I=1}^M\frac{Z_I}{|\boldsymbol{r_i} - \boldsymbol{R_I}|} + \\
    \frac{1}{2}\sum_{i=1}^N\sum_{j=i+1}^{N} \frac{1}{|\boldsymbol{r_i}-\boldsymbol{r_j}|} + \frac{1}{2}\sum_{I=1}^M\sum_{J=I+1}^M\frac{Z_A Z_J}{|\boldsymbol{R_I} - \boldsymbol{R_J}|}
\end{aligned}
\end{equation}

Assuming that the masses of the atomic nuclei are much larger than the mass of the electron, we can perform the \emph{Born-Oppenheimer} approximation. This means that we can write the total wavefunction as a product of an \emph{electronic} and a \emph{vibrational} (nuclear) wavefunction. The electronic Hamiltonian is written as:
\begin{equation}
    H_{\text{elec}} = -\sum_i \frac{\grad_i^2}{2} - \sum_{i,I}\frac{Z_I}{|\boldsymbol{r_i}-\boldsymbol{R_I}|} + \frac{1}{2}\sum_{i\neq j}\frac{1}{|\boldsymbol{r_i}-\boldsymbol{r_j}|}
\label{eq:electronic_hamiltonian}
\end{equation}
and in the second quantization representation, it can be reformulated as:
\begin{equation}
    H_{elec} = \sum_{pq}h_{pq}\hat{a}_p^\dagger \hat{a}_q + \frac{1}{2}\sum_{pqrs} h_{pqrs} \hat{a}_p^\dagger \hat{a}_q^\dagger \hat{a}_r\hat{a}_s
\label{eq:fermionic_hamiltonian}
\end{equation}
where the 1-body integrals:
\begin{equation}
    h_{pq} = \int \phi_p^*(\boldsymbol{r})\left(-\frac{1}{2}\grad^2 - \sum_I \frac{Z_I}{\boldsymbol{R_I}-\boldsymbol{r}}\right)\phi_q(\boldsymbol{r})d\boldsymbol{r}
\label{eq:one_electron_integrals}
\end{equation}
and 2-body integrals:
\begin{equation}
    h_{pqrs} = \int \frac{\phi_p^*(\boldsymbol{r_1})\phi_q^*(\boldsymbol{r_2})\phi_r(\boldsymbol{r_2})\phi_s(\boldsymbol{r_1})}{|\boldsymbol{r_1} - \boldsymbol{r_2}|}d\boldsymbol{r_1}d\boldsymbol{r_2}
\label{eq:two_electron_integrals}
\end{equation}
can be calculated efficiently using classical codes. The fermionic Hamiltonian in Eq. \eqref{eq:fermionic_hamiltonian} contains $\mathcal{O}(N_{SO}^4)$ terms, where $N_{SO}$ is the number of spin-orbitals considered and can be mapped into a qubit-Hamiltonian using the Jordan-Winger \cite{nielsen2005fermionic}, Bravyi-Kitaev \cite{mcardle2020quantum}, or Parity transformations \cite{seeley2012bravyi}. 

The number of spin orbitals can be reduced by considering an active space \cite{mcardle2020quantum, roos1980complete, helgaker2013molecular, ding2023quantum, de2023classical, battaglia2024general}. In several cases, the expected occupation number of some orbitals is either close to 0 or to 1, and as such, they can be assumed empty or occupied prior to the calculations. Then, the problem is reduced to the ambiguously occupied orbitals. Occupied orbitals are usually referred to as core, while the unoccupied orbitals are named virtual. In \cite{rossmannek2021quantum}, the authors described how to replace the one-electron integrals defined in \eqref{eq:one_electron_integrals} by an \emph{inactive} Fock operator and the active space contains an effective potential generated by the inactive electrons. Additionally, in \cite{takeshita2020increasing}, the authors showed how, by considering an active space with a constant number of qubits, one can achieve great accuracy by adding extra measurements.

For our experiments, we used the Unitary Coupled Cluster ansatz \cite{cooper2010benchmark, evangelista2019exact, ganzhorn2019gate} with double excitation (UCCD) and single-double excitations (UCCSD) but also problem-agnostic hardware-efficient ansatz families which are suitable for current devices \cite{kandala2017hardware}. The UCC ansatz is defined as:
\begin{equation}
    V(\boldsymbol{\theta}) = e^{T(\boldsymbol{\theta}) - T^\dagger(\boldsymbol{\theta})}
\end{equation}
where $T(\boldsymbol{\theta}) = T_1(\boldsymbol{\theta})+T_2(\boldsymbol{\theta}) + \ldots + T_n(\boldsymbol{\theta})$ is the excitation operator to order $n$. For instance, up to second order, the $T_i$ operators are defined as:
\begin{gather*}
    T_1(\boldsymbol{\theta}) = \sum_{i\in occ}\sum_{k \in virt} \theta_i^k \hat{a}_k^\dagger \hat{a}_i \\
    T_2(\boldsymbol{\theta}) = \frac{1}{2}\sum_{i,j\in occ} \sum_{k,l\in virt} \theta_{ij}^{kl} \hat{a}_l^\dagger \hat{a}_k^\dagger \hat{a}_j \hat{a}_i
\end{gather*}
where $occ$ denotes orbitals that are occupied in the Hartree-Fock state and $virt$ orbitals that are unoccupied. The application of a UCC-type ansatz for large molecules is rather challenging due to its large depth. Specifically, the implementation of UCCSD ansatz requires depth that scales as $\mathcal{O} ({N_{\text{occ}}\choose{2}}{N_{\text{virt}}\choose{2}}N_{\text{qubits}})$ where $N_{\text{occ}}, N_{\text{virt}}$ are the number of occupied and virtual orbitals respectively \cite{barkoutsos2018quantum}.

The system is initialized in the Hartree Fock state $\ket{\Phi_0}$ \cite{google2020hartree}, with $E_{HF} = \bra{\Phi_0}H_{\text{elec}}\ket{\Phi_0}$ where electron-electron correlations are neglected. Then, our goal is to identify the parameters $\boldsymbol{\theta^*}$ such that the loss function:
\begin{equation}
    \mathcal{L}(\boldsymbol{\theta}) = \bra{\phi(\boldsymbol{\theta})}H_{\text{elec}}\ket{\phi(\boldsymbol{\theta})}
\end{equation}
is minimized. 

Next, we discuss the several molecular systems for which we are interested in calculating their ground-state energy. We compare our methods with VarQITE, which uses the full QFIM and corresponds to the exact imaginary-time evolution for an ansatz family that is expressive enough to reach all intermediate ground states. We will use both estimators introduced in section Sec. \ref{sec:approximatin_qfisher_random} and show how RMITE compares to VarQITE when the QFIM is replaced by the aforementioned estimators.

\subsection{Technical Details}

For all our simulations, we used Qiskit Nature \cite{javadi2024quantum}, a quantum computing library suited for quantum algorithms for natural science problems, which is integrated with the PySCF \cite{sun2018pyscf} library that is used for quantum chemistry. For the molecular systems (i.e. $LiH$ and $H_2 O$) that we examined, we used the STO-3G basis. Additionally, to identify the active space, we utilized the FreezeCoreTransformer class from Qiskit \cite{javadi2024quantum}, which identifies the unique choice of active space given the number of electrons and the number of spatial orbitals.

\subsection{Lithium Hydride, $LiH$}

The first type of molecule that we choose to investigate is the \emph{Lithium Hydride $LiH$}. For this molecule, the Lithium (Li) and Hydrogen (H) atoms were allowed to have varying bond distances as seen in Figure \ref{fig:dissociation_profile}. For this molecule, we chose to use the average classical Fisher information matrix estimator given in Eq. \eqref{eq:average_classical_fisher_matrix}. The number of molecular orbitals is six, which corresponds to twelve spin orbitals. We can, however, reduce the number of required qubits by exploiting the active space \cite{rossmannek2021quantum} (i.e. using the FreezeCoreTransformer class), and reduce the molecular orbitals to five. As a result, we can map the Hamiltonian that describes the electronic structure of $LiH$ onto a ten-qubit Hamiltonian using the Bravyi-Kitaev transformation. 

\begin{figure}
\begin{tikzpicture}
\node (img)  {\includegraphics[scale=0.5]{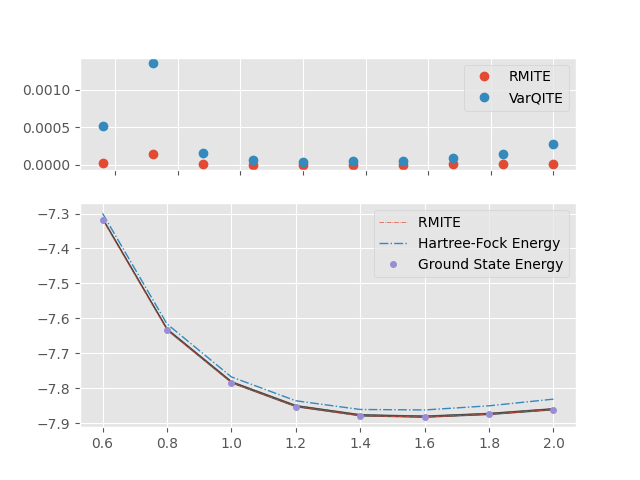}};
\node[below=of img, node distance=0cm, xshift=0.2cm, yshift=1.5cm] {\scriptsize Bond Distance [${\si{\angstrom}}$]};
\node[left=of img, node distance=0cm, rotate=90, xshift=-0.5cm, yshift=-1.2cm] {\scriptsize $\mathcal{L}(\boldsymbol{\theta})[\text{Ha}]$};
\node[left=of img, node distance=0cm, rotate=90, xshift=2.5cm, yshift=-1.2cm] {\scriptsize Relative Error};
\end{tikzpicture}
    \caption{Dissociation profile for the $LiH$ molecule. The top figure corresponds to the relative error in the calculation of the ground state for VarQITE and RMITE. The bottom figure corresponds to the energy returned by RMITE and the Hartree-Fock approximation. The former remains within chemical accuracy (black shaded area) close to the actual ground state.}
\label{fig:dissociation_profile}
\end{figure}
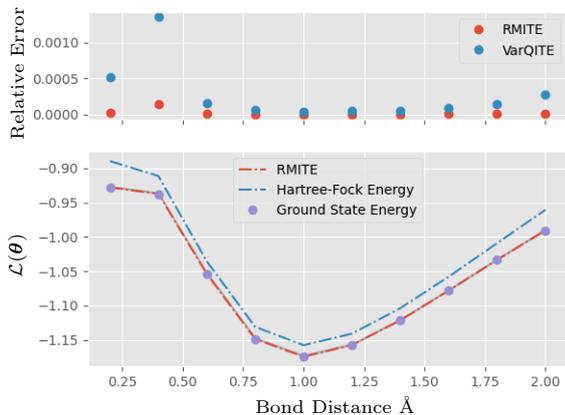

For $LiH$, we initially chose to use only a single random measurement per iteration. The random measurement is chosen to be a hardware-efficient ansatz circuit consisting of two layers of random Pauli rotations and nearest-neighbour connectivity. That is, in order to calculate the parameter dynamics, we use Eq. \eqref{eq:single_unitary_approximate_varqite}. The reason is that a single random classical Fisher information matrix can accurately approximate the geometry of the underlying quantum states. This was also discussed analytically in \cite{kolotouros2024random}. Furthermore, we employed a UCCSD ansatz family for which we initialize in the Hartree-Fock state (as discussed in Sec. \ref{sec:quantum_chemistry}).

In the top of Figure \ref{fig:dissociation_profile}, we illustrate the relative error in the ground state energy calculation of $LiH$ for different bond distances. The blue dots correspond to VarQITE \cite{mcardle2019variational}, where the full QFIM is used, while the red dots correspond to RMITE, in which the QFIM has been replaced by a random CFIM. Both methods were given 800 iterations with a timestep of $\delta t=0.01$. We can clearly see that the error of RMITE is significantly less than VarQITE. In the lower Figure \ref{fig:dissociation_profile}, we visualize how RMQITE compares to the HF approximation. Our method is able to remain always close (within chemical accuracy) to the true ground state.

In this case, our approach is able to offer large advantages in terms of resources, either quantum (i.e. quantum states prepared) or classical (on the number of matrices we have to store and updates in the parameters). As we can see in Figure \ref{fig:lih_resources_comparison}, RMITE outperforms VarQITE. On the left-hand subplot, we compare the number of optimizer iterations, i.e. the number of times we update the parameters of the parameterized quantum circuit. Moreover, the greatest advantage can be seen on the right-hand subplot of \ref{fig:lih_resources_comparison}, where we quantify the actual calls (different states we have to prepare) on the quantum computer. 

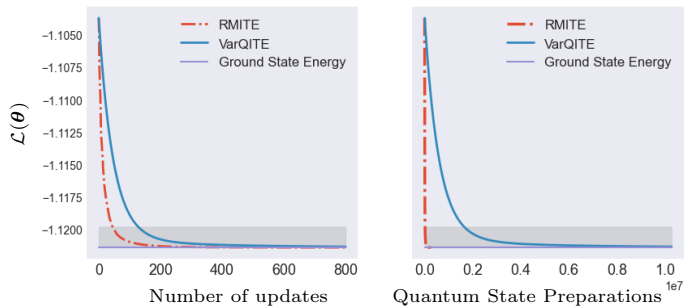
\begin{figure}
\centering
\begin{tikzpicture}
\node (img) {\includegraphics[scale=0.42]{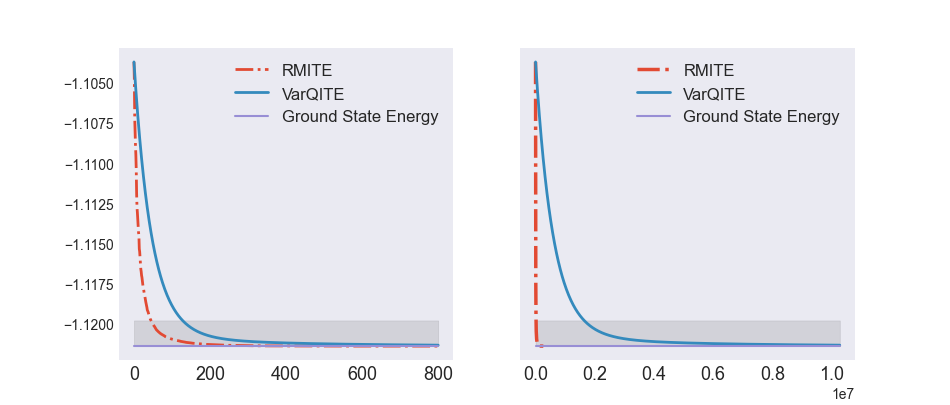}};
\node[left=of img, node distance=0cm, rotate=90, xshift=0.5cm, yshift=-1.5cm] {\scriptsize $\mathcal{L}(\boldsymbol{\theta})$[\text{Ha}]};
\node[below=of img, node distance=0cm, xshift=-1.8cm, yshift=1.3cm] {\scriptsize Number of updates};
\node[below=of img, node distance=0cm, xshift=2.0cm, yshift=1.3cm] {\scriptsize Quantum State Preparations};
\end{tikzpicture}
\caption{Comparison of RMITE and VarQITE on the number of updates (left figure) and quantum calls (right figure). The gray-shaded area corresponds to the chemical accuracy. The advantage of RMITE is illustrated in both figures.}
\label{fig:lih_resources_comparison}
\end{figure}

Moreover, we highlight the performance of our algorithm when multiple random classical Fisher information matrices are used per iteration, i.e. we investigate what happens when we increase $K$. Specifically, at each iteration, we calculate $K=5$ random CFIMs and use Eq. \eqref{eq:multi_unitary_approximate_varqite} to update the parameters. The results are illustrated in Figure \ref{fig:lih_multiple_cfims}. As it is clearly visualized, as we increase the number of sampled unitaries, we are able to remain close to the imaginary-time evolved states. In our case, we can achieve that by using significantly fewer resources than $\mathcal{O}(m^2)$ (see blue line). However, even a single random measurement results in a descent direction (see red line), offering a fast convergence. Thus, we can conclude that in the case of the average CFIM estimator, a single CFIM suffices to move in a descent direction, but multiple CFIMs are needed to remain close to the actual imaginary-time evolution.

\subsection{Water, $H_2 O$}

The second molecule we examine is water. For this molecule, the Oxygen (O) atom was at position $(0, 0, 0.115{\si{\angstrom}})$ and the two Hydrogen (H) atoms at the positions $(0, 0.754{\si{\angstrom}}, -0.459{\si{\angstrom}})$ and $(0, -0.754{\si{\angstrom}}, -0.459{\si{\angstrom}})$. This corresponds to a molecule with a bond length of approximately $0.958{\si{\angstrom}}$ and a bond angle close to $\ang{104.5}$. For this problem, we choose to use the estimator in Eq. \eqref{eq:quantum_from_random_measurements} in which the unitaries for the random measurements are drawn from the Clifford group. As a parameterized family of states, we used a hardware-efficient ansatz (as seen in Figure \ref{fig:hardware_efficient_ansatz}) with 4 layers. This architecture corresponds to a parameterized circuit of $m=100$ parameters. It is well-known that these hardware-efficient ansatz families suffer from the barren plateaux problem as the number of parameters increase (and are initialized at random), as both the variance and the gradient of the expectation value go exponentially to zero. However, they can provide a test bed for our method, showcasing its ability to work well with any chosen parameterized family of gates.

\begin{figure}
    \centering
    \includegraphics[scale=0.45]{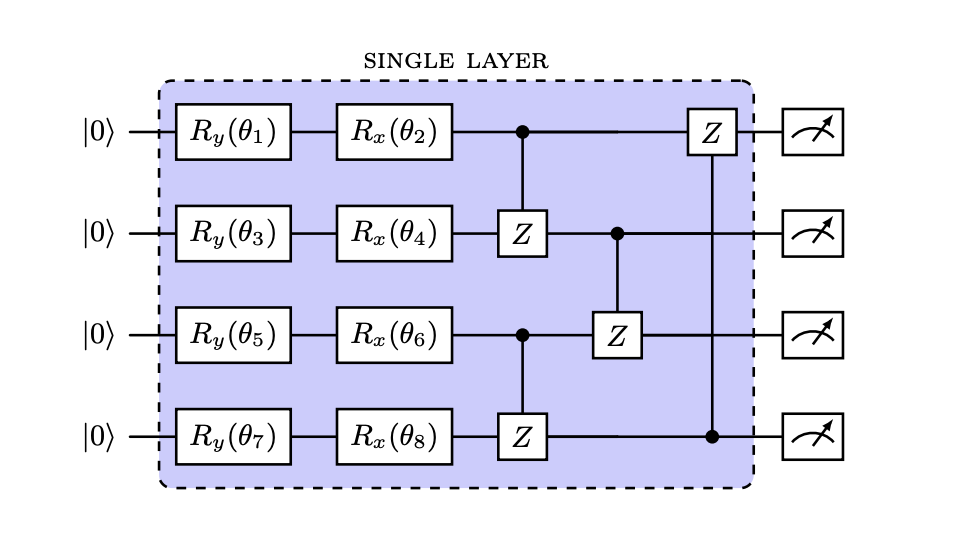}
    \caption{Example of a single-layer parameterized family of gates used for the $H_2O$ experiments.}
    \label{fig:hardware_efficient_ansatz}
\end{figure}

For the estimator in Eq. \eqref{eq:quantum_from_random_measurements}, we used $K=10$ and $K=20$ random unitaries per iteration to show how one can achieve a very good approximation with significantly less calls on the quantum processor. The reason that we increased the number of random unitaries compared to the average CFIM case is that this estimator requires more samples to achieve a good approximation of the QFIM (but still much less resources than those required for the QFIM). As it can be clearly visualized in Figure \ref{fig:h2o_ground_state_preparation}, our algorithms offers a valuable advantage over VarQITE since the number of calls on the quantum processor is reduced significantly.

Overall, we can see that the advantages of our method are two-faced. First of all, when we choose a small number of repetitions $K$, the algorithm resembles a quantum information-theoretic optimization algorithm. That is, we get information about whether certain parameters do change the underlying quantum state or not, similar to \cite{kolotouros2024random, stokes2020quantum}. On the other hand, as $K$ increases, our estimator approximates even better the QFIM. As such, the method transforms from an information-theoretic optimization technique to an approximator of the QFIM, and thus, RMITE resembles VarQITE.

\begin{figure}
\begin{tikzpicture}
\node (img)  {\includegraphics[scale=0.53]{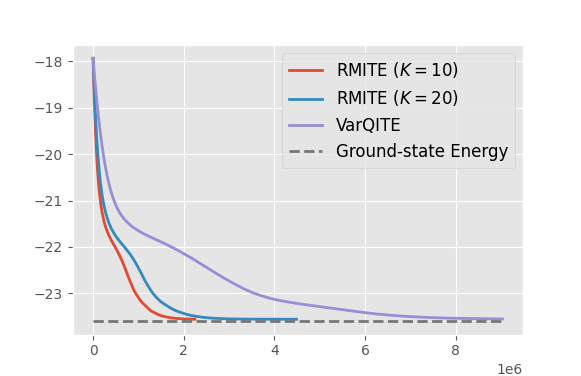}};
\node[below=of img, node distance=0cm, yshift=1.3cm] (xlabel1) {\scriptsize Quantum State Preparations};
\node[left=of img, node distance=0cm, rotate=90, yshift=-1.1cm, xshift=0.6cm] {\scriptsize $\mathcal{L}(\boldsymbol{\theta})$[\text{Ha}]};
\end{tikzpicture}
\caption{Performance of RMITE (with Eq. \eqref{eq:quantum_from_random_measurements} as the estimator) compared to VarQITE on the task of preparing the ground state of $H_2O$.}
\label{fig:h2o_ground_state_preparation}
\end{figure}


\begin{figure}
\begin{tikzpicture}
\node (img)  {\includegraphics[scale=0.53]{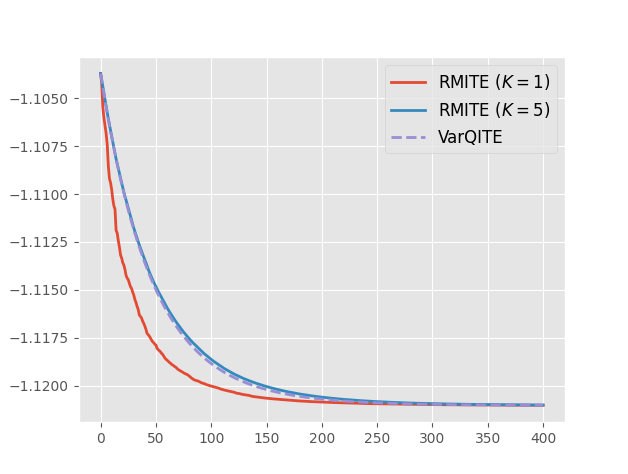}};
\node[below=of img, node distance=0cm, yshift=1.3cm] (xlabel1) {\scriptsize Iterations};
\node[left=of img, node distance=0cm, rotate=90, yshift=-0.9cm, xshift=0.6cm] {\scriptsize $\mathcal{L}(\boldsymbol{\theta})$[\text{Ha}]};
\end{tikzpicture}
\caption{Performance of imaginary-time evolution using a single random measurement, multiple random measurements, or the quantum Fisher information matrix per iteration and the estimator in Eq. \eqref{eq:cfisher_haar}. Increasing the number of random measurements improves the approximation to the imaginary-time evolution, while fewer random measurements result in rapid descent.}
\label{fig:lih_multiple_cfims}
\end{figure}

\section{Discussion}
\label{sec:discussion}

In this paper, we introduced two novel quantum algorithms that bring imaginary-time evolution a step closer to a practical implementation on a quantum computer. We achieve this by minimizing the number of calls on the quantum processor while at the same time using the same number of classical resources.

Quantum imaginary-time evolution can be performed using a hybrid quantum/classical approach. Specifically, the user specifies a parameterized architecture on a quantum computer and its parameters are iteratively updated so that the parameterized state remains as close as possible to the imaginary-evolved quantum state. The parameter dynamics are found by solving a differential equation that is derived by measuring $\Theta(m^2)$ quantum states at each iteration. The need for $\Theta(m^2)$ quantum states comes from the fact that the partial differential equation that governs the dynamics of the parameters requires the calculation of the quantum Fisher information matrix (QFIM), an object that describes the underlying geometry of parameterized quantum states. This imposes a bottleneck in the practicality of this method as the number of parameters becomes large.

In this paper, we investigated how one can approximate the QFIM using tools from the random-measurement theory \cite{elben2023randomized}. We showed that the QFIM can be reconstructed to user-dependent precision, using $\Theta(Km)$ quantum states, where $m$ is the number of parameters in the parameterized quantum circuit and $K$ is user-specified and, in practice, much smaller than $m$. To approximate the QFIM, one has to rotate the parameterized state by a random unitary $U$ and then measure the state on the computational basis.

Our first result is inspired by \cite{elben2019statistical}, where the authors showed how the fidelity of two quantum states can be calculated using random measurements. By exploiting their result as well as some tools from random matrix theory \cite{mele2024introduction}, we proposed a quantity that can estimate the quantum Fisher information matrix by randomly rotating a quantum state by a unitary that is sampled from a 2-design and then measuring the state on the computational basis. In our experiments, we chose unitaries to be sampled from the Clifford group.

Then, we proposed that the QFIM can also be approximated by post-processing classical Fisher information matrices (CFIMs), where each CFIM requires quadratically less quantum resources than the QFIM. The latter approximator converges faster (in terms of error in the $l_2$ norm) to the QFIM when the measurement is random. As we prove, even a single random measurement points in the descent direction for a sufficiently small step size. Finally, based on our findings, we propose two different quantum/classical algorithms where we benchmark on quantum chemistry problems, where the task is to prepare their ground state. In all experiments, our algorithms converge significantly faster than VarQITE, requiring fewer calls on the quantum processor while at the same time reaching ground state approximations within chemical accuracy.

Our method opens up many directions of research for testing the limit of practicality of already known quantum algorithms when certain computationally expensive quantities (such as the QFIM) are estimated (up to some error) by computationally cheaper objects. Techniques such as our proposed algorithm showcase their practicality as quantum-inspired classical approximation methods, on top of the advantages that were previously mentioned. A user can classically prepare and store ansatz states, and by performing $K$ measurements on the stored state, they can approximate the QFIM and, eventually, identify the direction that follows imaginary-time evolution. 

As we showed in Appendix \ref{appendix:error_in_approximation}, our estimators require less quantum resources if we allow for a specific tolerance in the approximation. However, the user still needs to estimate the sum (over an exponentially large set of outcomes) of products of probability distributions and their derivatives. The questions of whether this can be done efficiently or if considering a smaller subset of the outcomes still provides a good estimator remain open. Identifying whether this can be done efficiently in a quantum computer will determine whether our proposed method is suited as a quantum algorithm or a quantum-inspired algorithm (since the quantum resources are less than the QFIM, providing an advantage in that sense).

Finally, another interesting research direction is to theoretically quantify the sample complexity $K$ on both estimators. Different choices for $K$ may be needed to achieve a desired accuracy $\epsilon$ (error from the QFIM) for different parameterized quantum states. Thus, it is essential to understand how the geometry of the underlying quantum states is connected to the sampling requirements of our algorithm.

\newpage 

\bibliographystyle{unsrt}
\bibliography{References}

\onecolumngrid

\appendix

\section{Random Measurements and QFIM}
\label{appendix:random_measurements}

Randomized measurements constitute a powerful tool that has been exploited for several different applications throughout the quantum computing literature \cite{elben2023randomized, elben2019statistical, notarnicola2023randomized, huang2020predicting, brydges2019probing}. We begin by recalling a few standard notions on random operators acting on our space of parametrized qubits, referring to  \cite{roberts2017chaos, mele2024introduction} for a comprehensive introduction. 

\begin{defi}
(Haar Measure) \cite{mele2024introduction}. The Haar measure on the unitary group $U(d)$ is the unique probability measure $\mu_H$ \cite{simon1996representations} that is both \emph{left} and \emph{right} invariant over the group $U(d)$, i.e. for all integrable function $f: U(d) \rightarrow \mathcal{L}(\mathbb{C}^d)$ and for all $V\in U(d)$ we have:
\begin{equation}
    \int_{U(d)} f(U)d\mu_H(U) = \int_{U(d)} f(UV)d\mu_H(U) = \int_{U(d)} f(VU)d\mu_H(U)
\label{eq:haar_measure}
\end{equation}
\end{defi}
In this paper, we will denote the integral of a function $f(U)$ over the Haar measure as the expected value of $f(U)$ with respect to the probability measure $\mu_H$, denoted as $\mathbb{E}_{U\sim \mu_H} [f(U)]$:
\begin{equation}
   \mathbb{E}_{U\sim \mu_H} [f(U)] := \int_{U(d)} f(U)d\mu_H(U)
\end{equation}

A quantity that will play a very important role in our analysis is the $k$-moment operator, with $k \in \mathbb{N}$ (or else the $k$-fold twirl).

\begin{defi}
    ($k$-moment operator). The $k$-moment operator, with respect to the probability measure $\mu_H$, is defined as $\mathcal{M}_k : \mathcal{L}((\mathbb{C}^d)^{\otimes k}) \rightarrow \mathcal{L}((\mathbb{C}^d)^{\otimes k})$:
    \begin{equation}
        \mathcal{M}_{\mu_H}^{(k)}(O) := \mathbb{E}_{U\sim \mu_H} [U^{\otimes k} O U^{\dagger \otimes k}]
    \end{equation}
    for all operators $O \in \mathcal{L}((\mathbb{C}^d)^{\otimes k})$.
\label{definition:k_moment_operator}
\end{defi}

As it turns out, there are tools that will allow us to calculate the $k$-th moment operators. Specifically, the moment operator defined in Definition \ref{definition:k_moment_operator} is the orthogonal projector onto the commutant $Comm(U(d),k)$. The commutant is defined below.

\begin{defi}
    (Commutant). Given $S \subseteq \mathcal{L}(\mathbb{C}^d)$, we define its $k$-th order commutant as:
    \begin{equation}
        Comm(S,k) := \{ A \in \mathcal{L}((\mathbb{C}^d)^{\otimes k}): [A, B^{\otimes k}] = 0 \; \; \forall B\in S\}
    \end{equation}
\end{defi}

As it can be easily seen, a set of operators that commute with every unitary $U^{\otimes k}$ are the permutation operators. These are defined as:

\begin{defi}
    (Permutation operators). Given $\pi \in S_k$ an element of the symmetric group $S_k$, we define the permutation matrix $V_d(\pi)$ to be the unitary matrix that satisfies:
    \begin{equation}
        V_d(\pi) \ket{\psi_1}\otimes \ldots \otimes \ket{\psi_k} = \ket{\psi_{\pi^{-1}(1)}}\otimes \ldots \otimes \ket{\psi_{\pi^{-1}(k)}}
    \end{equation}
    for all $\ket{\psi_1},\ldots, \ket{\psi_k}\in \mathbb{C}^d$
\end{defi}

A well-celebrated result is the Schur-Weyl duality \cite{roberts2017chaos}, that states that the image of the $k$-moment operator is spanned by the permutation operators. As such, we can calculate the first and second moments of an operator $O \in \mathcal{L}(\mathbb{C}^d)$ as:

\begin{gather}
             \mathbb{E}_{U\sim \mu_H} [U O U^\dagger] = \frac{\Tr(O)}{d}I \\
              \mathbb{E}_{U\sim \mu_H} [U^{\otimes 2} O U^{\dagger \otimes 2}] = \frac{\Tr(O) - d^{-1}\Tr(\mathbb{S}O)}{d^2 - 1} \mathbb{I} + \frac{\Tr(\mathbb{S}O) - d^{-1}\Tr(O)}{d^2 - 1} \mathbb{S}
\label{eq:haar_moments}
\end{gather}
where $I, \mathbb{I}$ correspond to the identity operators on $\mathbb{C}^d$ and $(\mathbb{C}^d)^{\otimes 2}$ respectively and $\mathbb{S}$ is the SWAP operator defined as:
\begin{equation}
    \mathbb{S}(\ket{\psi_1}\otimes \ket{\psi_2})= \ket{\psi_2} \otimes \ket{\psi_1}
\end{equation}
Our starting point originates from \cite{brydges2019probing,elben2019statistical}, where the authors showed that the fidelity between two quantum states $\rho_1,\rho_2$ can be calculated using the following Theorem.

\begin{theorem}
    (Fidelity of two quantum states \cite{elben2019statistical}) Consider two quantum states $\rho_1$ and $\rho_2$ on $n$ qubits in Hilbert space $\mathcal{H}$ of dimension $\mathcal{D} = 2^n$. For global random unitaries $U$, the overlap between the quantum states is given by:
    \begin{equation}
        \Tr [\rho_1 \rho_2] = 2^n \sum_{\boldsymbol{s},\boldsymbol{s'}} (-2^n)^{-D_G[\boldsymbol{s},\boldsymbol{s'}]}\bra{\boldsymbol{s}}\bra{\boldsymbol{s'}}\mathcal{M}_{\mu_H}^{(2)}(\rho_1 \otimes \rho_2)\ket{\boldsymbol{s}}\ket{\boldsymbol{s'}}
    \label{eq:fidelity_with_random_measurements}
    \end{equation}
    where the global Hamming distance $D_G$ is defined as:
    \begin{equation}
        D_G[\boldsymbol{s}, \boldsymbol{s'}] = \begin{cases}
            0 \; \text{if $\boldsymbol{s} = \boldsymbol{s'}$}\\
            1 \; \text{if $\boldsymbol{s} \neq \boldsymbol{s'}$}
        \end{cases}
    \label{eq:global_hamming_distance}
    \end{equation}
    and $\mathcal{M}_{\mu_H}^{(k)}(\cdot) := \mathbb{E}_{U\sim \mu_H} [U^{\otimes k} (\cdot) U^{\dagger \otimes k}]$ is the $k$-th moment operator.
\end{theorem}

In our case, we work with \emph{parameterized quantum states} consisting of a total of $m$ parameters $\boldsymbol{\theta} = (\theta_1, \theta_2, \ldots, \theta_m)$. Specifically, let $\rho_1 := \rho(\boldsymbol{\theta})$ and $\rho_2 := \rho(\boldsymbol{\theta}+\boldsymbol{\epsilon})$ and let also $\rho_1, \rho_2$ be \emph{pure}. As such, $\rho_1$ and $\rho_2$ can be written as:
\begin{equation}
\begin{gathered}
    \rho_1 = \rho(\boldsymbol{\theta}) = \ketbra{\phi(\boldsymbol{\theta})} \\
    \rho_2 = \rho(\boldsymbol{\theta} + \boldsymbol{\epsilon}) = \ketbra{\phi(\boldsymbol{\theta}+\boldsymbol{\epsilon})}
\end{gathered}
\end{equation}
 
  We can express the elements of the moment operator in the computation basis as:
\begin{equation}
\begin{gathered}
    \bra{\boldsymbol{s},\boldsymbol{s'}}\mathcal{M}_{\mu_H}^{(2)}(\rho(\boldsymbol{\theta}) \otimes \rho(\boldsymbol{\theta}+\boldsymbol{\epsilon})) \ket{\boldsymbol{s},\boldsymbol{s'}} = \int_{U(d)} d\mu_H(U) \bra{\boldsymbol{s}}\bra{\boldsymbol{s'}}U^{\otimes 2}\rho(\boldsymbol{\theta})\otimes \rho(\boldsymbol{\theta}+\boldsymbol{\epsilon}) U^{\dagger \otimes 2} \ket{\boldsymbol{s}}\ket{\boldsymbol{s'}} \\
    \int_{U(d)} d\mu_H(U) \bra{\boldsymbol{s}}U\rho(\boldsymbol{\theta}) U^{\dagger} \ket{\boldsymbol{s}}\bra{\boldsymbol{s'}}U\rho(\boldsymbol{\theta}+\boldsymbol{\epsilon}) U^{\dagger}\ket{\boldsymbol{s'}} = \int_{U(d)} d\mu_H(U) \Tr[\rho(\boldsymbol{\theta})U^\dagger \Pi_{\boldsymbol{s}} U]\Tr[\rho(\boldsymbol{\theta}+\boldsymbol{\epsilon})U^\dagger \Pi_{\boldsymbol{s'}} U] \\
    = \int_{U(d)} d\mu_H (U) \Tr[\rho(\boldsymbol{\theta}) \otimes\rho(\boldsymbol{\theta}+\boldsymbol{\epsilon})  U^{\dagger \otimes 2}(\Pi_{\boldsymbol{s}} \otimes \Pi_{\boldsymbol{s'}}) U^{\otimes 2}] = \mathbb{E}_{U\sim \mu_H}\big[p_{\boldsymbol{s}}^U(\boldsymbol{\theta}) p_{\boldsymbol{s'}}^U (\boldsymbol{\theta}+\boldsymbol{\epsilon})\big]
\end{gathered}
\end{equation}
where $p_{\boldsymbol{s}}^U = \Tr[\rho(\boldsymbol{\theta})U^\dagger \Pi_{\boldsymbol{s}} U]$. As such we can rewrite Eq. \eqref{eq:fidelity_with_random_measurements} as:
\begin{equation}
\begin{gathered}
    \Tr [\rho(\boldsymbol{\theta}) \rho(\boldsymbol{\theta} + \boldsymbol{\epsilon})] = 2^n \sum_{\boldsymbol{s}, \boldsymbol{s'}} (-2^n)^{-D_G[\boldsymbol{s}, \boldsymbol{s'}]} \Tr \Big[\rho(\boldsymbol{\theta}) \otimes \rho(\boldsymbol{\theta}+\boldsymbol{\epsilon})\mathcal{M}_{\mu_H}^{(2)}(\Pi_{\boldsymbol{s}} \otimes \Pi_{\boldsymbol{s'}}) \Big] \\
    =2^n \sum_{\boldsymbol{s}, \boldsymbol{s'}} (-2^n)^{-D_G[\boldsymbol{s}, \boldsymbol{s'}]} \mathbb{E}_{U\sim \mu_H}\big[p_{\boldsymbol{s}}^U(\boldsymbol{\theta}) p_{\boldsymbol{s'}}^U (\boldsymbol{\theta}+\boldsymbol{\epsilon})\big]
\end{gathered}
\end{equation}
By expanding the sum, the fidelity between the two states can be written as:
\begin{equation}
\begin{gathered}
    \Tr [\rho(\boldsymbol{\theta}) \rho(\boldsymbol{\theta}+\boldsymbol{\epsilon})] = 2^n \sum_{\boldsymbol{s},\boldsymbol{s'}} (-2^n)^{-D_G[\boldsymbol{s},\boldsymbol{s'}]}\mathbb{E}_{U\sim \mu_H}\big[p_{\boldsymbol{s}}^U(\boldsymbol{\theta}) p_{\boldsymbol{s'}}^U (\boldsymbol{\theta}+\boldsymbol{\epsilon})\big]  \\
     = 2^n \sum_{\boldsymbol{s}} (-2^n)^{-D_G[\boldsymbol{s},\boldsymbol{s}]}\mathbb{E}_{U\sim \mu_H}\big[p_{\boldsymbol{s}}^U(\boldsymbol{\theta}) p_{\boldsymbol{s'}}^U (\boldsymbol{\theta}+\boldsymbol{\epsilon})\big] +     2^n \sum_{\substack{\boldsymbol{s},\boldsymbol{s'} \\ \boldsymbol{s} \neq \boldsymbol{s'}}} (-2^n)^{-D_G[\boldsymbol{s},\boldsymbol{s'}]}\mathbb{E}_{U\sim \mu_H}\big[p_{\boldsymbol{s}}^U(\boldsymbol{\theta}) p_{\boldsymbol{s'}}^U (\boldsymbol{\theta}+\boldsymbol{\epsilon})\big]\\
     =  2^n \sum_{\boldsymbol{s}} \mathbb{E}_{U\sim \mu_H}\big[p_{\boldsymbol{s}}^U(\boldsymbol{\theta}) p_{\boldsymbol{s'}}^U (\boldsymbol{\theta}+\boldsymbol{\epsilon})\big] - \sum_{\substack{\boldsymbol{s},\boldsymbol{s'} \\ \boldsymbol{s} \neq \boldsymbol{s'}}}\mathbb{E}_{U\sim \mu_H}\big[p_{\boldsymbol{s}}^U(\boldsymbol{\theta}) p_{\boldsymbol{s'}}^U (\boldsymbol{\theta}+\boldsymbol{\epsilon})\big]
\end{gathered}
\end{equation} 
where we used Eq. \eqref{eq:global_hamming_distance}. Consider now the infidelity between two quantum states as defined in Eq. \eqref{eq:infidelity}
\begin{equation}
    d_F\Big(\rho(\boldsymbol{\theta}),\rho( \boldsymbol{\theta}+\boldsymbol{\epsilon})\Big) = 1 - \Tr [\rho(\boldsymbol{\theta}) \rho(\boldsymbol{\theta}+\boldsymbol{\epsilon})]
\end{equation}
If we Taylor expand the infidelity around $\boldsymbol{\epsilon} = \boldsymbol{0}$, then we have:
\begin{equation*}
    \begin{gathered}
       1 - \Tr[\rho(\boldsymbol{\theta}) \rho(\boldsymbol{\theta} + \boldsymbol{\epsilon})] 
        = 1 -\Tr[\rho(\boldsymbol{\theta})\rho(\boldsymbol{\theta})] - \sum_{i=1}^m \frac{\partial}{\partial \epsilon_i}\Tr[\rho(\boldsymbol{\theta}) \rho(\boldsymbol{\theta} + \boldsymbol{\epsilon})]\Big{|}_{\boldsymbol{\epsilon} = 0} \epsilon_i \\
         - \frac{1}{2}\sum_{i,j=1}^m \frac{\partial}{\partial \epsilon_i\partial \epsilon_j}\Tr[\rho(\boldsymbol{\theta}) \rho(\boldsymbol{\theta} + \boldsymbol{\epsilon})]\Big{|}_{\boldsymbol{\epsilon} = 0} \epsilon_i \epsilon_j + \mathcal{O}(\norm{\boldsymbol{\epsilon}}_1^3)
    \end{gathered}
\end{equation*}
where the partial derivatives at $\boldsymbol{\epsilon}= \boldsymbol{0}$ are zero, since $\Tr[\rho(\boldsymbol{\theta})\rho(\boldsymbol{\theta}+\boldsymbol{\epsilon})]$ is maximized at $\boldsymbol{\epsilon} = \boldsymbol{0}$. As such, the infidelity, can be expressed as:
\begin{equation}
    d_F\Big(\rho(\boldsymbol{\theta}),\rho( \boldsymbol{\theta}+\boldsymbol{\epsilon})\Big) = \frac{1}{4}\boldsymbol{\epsilon}^T \mathcal{F}_Q (\boldsymbol{\theta}) \boldsymbol{\epsilon} + \mathcal{O}(\norm{\boldsymbol{\epsilon}}_1^3)
\end{equation}
where for small $\boldsymbol{\epsilon}$, the higher-order terms can be neglected. The matrix $\mathcal{F}_Q (\boldsymbol{\theta})$ is the \emph{quantum Fisher information matrix}, defined as:
\begin{equation}
    \mathcal{F}_Q(\boldsymbol{\theta}) = - 2 \grad^2 \Tr[\rho(\boldsymbol{\theta})\rho(\boldsymbol{\theta} + \boldsymbol{\epsilon})]\Big{|}_{\boldsymbol{\epsilon} = 0}
\end{equation}
Thus, a matrix element $[F_Q(\boldsymbol{\theta})]_{ij}$ can be written as:
\begin{equation}
    [\mathcal{F}_Q(\boldsymbol{\theta})]_{ij} =  -2^{n+1} \sum_{\boldsymbol{s}} \mathbb{E}_{U\sim \mu_H}\Bigg[p_{\boldsymbol{s}}^U(\boldsymbol{\theta})  \frac{\partial^2 p_{\boldsymbol{s}}^U(\boldsymbol{\theta}+\boldsymbol{\epsilon})}{\partial \epsilon_i \partial \epsilon_j}\Big{|}_{\boldsymbol{\epsilon}=0}\Bigg] +2 \sum_{\substack{\boldsymbol{s},\boldsymbol{s'} \\\boldsymbol{s} \neq \boldsymbol{s'}}}\mathbb{E}_{U\sim \mu_H}\Bigg[p_{\boldsymbol{s}}^U(\boldsymbol{\theta}) \frac{\partial^2 p_{\boldsymbol{s'}}^U(\boldsymbol{\theta}+\boldsymbol{\epsilon})}{\partial \epsilon_i \partial \epsilon_j}\Big{|}_{\boldsymbol{\epsilon}=0}\Bigg]
\label{eq:qfisher_elements2}
\end{equation}
If we focus at the second term of Eq \eqref{eq:qfisher_elements2}, we notice that:
\begin{equation}
\begin{gathered}
    \sum_{\substack{\boldsymbol{s},\boldsymbol{s'} \\ \boldsymbol{s} \neq \boldsymbol{s'}}}p_{\boldsymbol{s}}^U(\boldsymbol{\theta}) \frac{\partial^2 p_{\boldsymbol{s'}}^U(\boldsymbol{\theta}+\boldsymbol{\epsilon})}{\partial \epsilon_i \partial \epsilon_j} = \sum_{\boldsymbol{s}} p_{\boldsymbol{s}}^U(\boldsymbol{\theta}) \sum_{\boldsymbol{s'}\neq \boldsymbol{s}} \frac{\partial^2 p_{\boldsymbol{s'}}^U(\boldsymbol{\theta}+\boldsymbol{\epsilon})}{\partial \epsilon_i \partial \epsilon_j}\\
    \sum_{\boldsymbol{s}} p_{\boldsymbol{s}}^U(\boldsymbol{\theta}) \Bigg(\sum_{\boldsymbol{s'}}\frac{\partial^2 p_{\boldsymbol{s'}}^U(\boldsymbol{\theta}+\boldsymbol{\epsilon})}{\partial \epsilon_i \partial \epsilon_j} - \frac{\partial^2 p_{\boldsymbol{s}}^U(\boldsymbol{\theta}+\boldsymbol{\epsilon})}{\partial \epsilon_i \partial \epsilon_j}\Bigg) = -  \sum_{\boldsymbol{s}} p_{\boldsymbol{s}}^U(\boldsymbol{\theta}) \frac{\partial^2 p_{\boldsymbol{s}}^U(\boldsymbol{\theta}+\boldsymbol{\epsilon})}{\partial \epsilon_i \partial \epsilon_j}
\end{gathered}    
\end{equation}
where we used the fact that:
\begin{equation}
    \sum_{\boldsymbol{s'}} \frac{\partial^2 p_{\boldsymbol{s'}}^U(\boldsymbol{\theta}+\boldsymbol{\epsilon})}{\partial \epsilon_i \partial \epsilon_j} = \frac{\partial^2}{\partial\epsilon_i \partial\epsilon_j} \sum_{\boldsymbol{s'}}p_{\boldsymbol{s'}}^U(\boldsymbol{\theta}+\boldsymbol{\epsilon}) =0
\label{eq:partial_zero}
\end{equation}
Thus, the QFIM elements can be expressed as:
\begin{equation}
    [\mathcal{F}_Q(\boldsymbol{\theta})]_{ij} =  -(2^{n+1} + 2) \sum_{\boldsymbol{s}} \mathbb{E}_{U\sim \mu_H}\Bigg[p_{\boldsymbol{s}}^U(\boldsymbol{\theta})  \frac{\partial^2 p_{\boldsymbol{s}}^U(\boldsymbol{\theta}+\boldsymbol{\epsilon})}{\partial \epsilon_i \partial \epsilon_j}\Big{|}_{\boldsymbol{\epsilon}=0}\Bigg] = -(2^{n+1} + 2) \sum_{\boldsymbol{s}} \mathbb{E}_{U\sim \mu_H}\Bigg[p_{\boldsymbol{s}}^U(\boldsymbol{\theta})  \frac{\partial^2 p_{\boldsymbol{s}}^U(\boldsymbol{\theta})}{\partial \theta_i \partial \theta_j}\Bigg]
\label{eq:qfisher_elements3}
\end{equation}

The calcuation of the quantum Fisher information using Eq. \eqref{eq:qfim_elements} is impractical. The reason is that it requires the calculation of the Hessian of the outcome probabilities, which in general requires $O(m^2)$ quantum states to estimate it. However, we can calculate that:
\begin{equation}
\begin{gathered}
    \sum_{\boldsymbol{s}} \mathbb{E}_{U\sim \mu_H}\Bigg[\frac{\partial p_{\boldsymbol{s}}^U(\boldsymbol{\theta})}{\partial \theta_i}\frac{\partial p_{\boldsymbol{s}}^U(\boldsymbol{\theta})}{\partial \theta_j}\Bigg] =     \sum_{\boldsymbol{s}} \mathbb{E}_{U\sim \mu_H}\Bigg[ \Tr [\frac{\partial \rho(\boldsymbol{\theta})}{\partial \theta_i} U^\dagger \Pi_{\boldsymbol{s}} U]\Tr [\frac{\partial \rho(\boldsymbol{\theta})}{\partial \theta_j} U^\dagger \Pi_{\boldsymbol{s}} U]\Bigg] \\
    =  \sum_{\boldsymbol{s}} \mathbb{E}_{U\sim \mu_H} \Bigg[ \Tr [\frac{\partial \rho(\boldsymbol{\theta})}{\partial \theta_i} \otimes \frac{\partial \rho(\boldsymbol{\theta})}{\partial \theta_j} U^{\dagger \otimes 2} \Pi_{\boldsymbol{s}}^{\otimes 2} U^{\otimes 2}]\Bigg] =  \frac{1}{2^{2n} - 1}\sum_{\boldsymbol{s}} \Tr [\frac{\partial \rho(\boldsymbol{\theta})}{\partial \theta_i} \otimes \frac{\partial \rho(\boldsymbol{\theta})}{\partial \theta_j}\Big[ \Big(1-\frac{1}{2^n}\Big)\mathbb{I} + \Big(1-\frac{1}{2^n}\Big)\mathbb{S}\Big] ] \\
    = \frac{2^n-1}{2^{3n} - 2^n} \sum_{\boldsymbol{s}} \Tr [\frac{\partial \rho(\boldsymbol{\theta})}{\partial \theta_i}] \Tr [\frac{\partial \rho(\boldsymbol{\theta})}{\partial \theta_j}] + \frac{2^n-1}{2^{3n} - 2^n} \sum_{\boldsymbol{s}} \Tr [\frac{\partial \rho(\boldsymbol{\theta})}{\partial \theta_i}\frac{\partial \rho(\boldsymbol{\theta})}{\partial \theta_j}] \\
    = \frac{2^n-1}{2^{3n} - 2^n} \sum_{\boldsymbol{s}} \Tr [\frac{\partial \rho(\boldsymbol{\theta})}{\partial \theta_i}\frac{\partial \rho(\boldsymbol{\theta})}{\partial \theta_j}] = \frac{1}{2^n + 1} \Tr [\frac{\partial \rho(\boldsymbol{\theta})}{\partial \theta_i}\frac{\partial \rho(\boldsymbol{\theta})}{\partial \theta_j}]
\label{eq:calculation_inner_product_derivatives_probabilities}
\end{gathered}
\end{equation}
where in the second line we used Eq. \eqref{eq:haar_moments} and we also used the fact that $\Tr[ \frac{\partial \rho(\boldsymbol{\theta})}{\partial \theta_i}] = \frac{\partial}{\partial \theta_i} \Tr [\rho(\boldsymbol{\theta})] = 0$ and $\Tr[A\otimes B \mathbb{S}] = \Tr[AB]$. Now, we can use the fact that:
\begin{equation}
    \frac{ \partial^2 \rho^2(\boldsymbol{\theta})}{\partial \theta_i \partial \theta_j} = 2\frac{\partial \rho(\boldsymbol{\theta})}{\partial \theta_i } \frac{\partial \rho(\boldsymbol{\theta})}{\partial \theta_j} + 2 \rho(\boldsymbol{\theta}) \frac{\partial^2 \rho(\boldsymbol{\theta})}{\partial \theta_i \partial \theta_j}
\end{equation}
Thus, substituting the above equation in Eq. \eqref{eq:calculation_inner_product_derivatives_probabilities} we get:
\begin{equation}
    \sum_{\boldsymbol{s}} \mathbb{E}_{U\sim \mu_H}\Bigg[\frac{\partial p_{\boldsymbol{s}}^U(\boldsymbol{\theta})}{\partial \theta_i}\frac{\partial p_{\boldsymbol{s}}^U(\boldsymbol{\theta})}{\partial \theta_j}\Bigg] = \frac{1}{2(2^n+1)} \Tr [\frac{\partial^2 \rho^2(\boldsymbol{\theta})}{\partial \theta_i \partial \theta_j}] - \frac{1}{(2^n+1)} \Tr [\rho(\boldsymbol{\theta}) \frac{\partial^2 \rho(\boldsymbol{\theta})}{\partial \theta_i \partial \theta_j}] = \frac{[\mathcal{F}_Q]_{ij}}{2(2^n+1)}
\end{equation}
where we used the fact that $\Tr[ \frac{\partial^2 \rho(\boldsymbol{\theta})}{\partial \theta_i \partial \theta_j}] = \frac{\partial^2}{\partial \theta_i \partial \theta_j} \Tr [\rho(\boldsymbol{\theta})] = 0$. As such, we were able to prove that the matrix elements of the quantum Fisher information matrix can be written as product of first-order derivatives:
\begin{equation}
    [\mathcal{F}_Q(\boldsymbol{\theta})]_{ij} = 2(2^n + 1)    \sum_{\boldsymbol{s}} \mathbb{E}_{U\sim \mu_H}\Bigg[\frac{\partial p_{\boldsymbol{s}}^U(\boldsymbol{\theta})}{\partial \theta_i}\frac{\partial p_{\boldsymbol{s}}^U(\boldsymbol{\theta})}{\partial \theta_j}\Bigg]
\end{equation}
As a result, we proved that the quantum Fisher information matrix can be approximated as the average (over the Haar distribution) of a quantity that requires $\mathcal{O}(m)$ quantum states and $O(m^2)$ classical memory to store the matrix. As we see in our numerical experiments in Fig. \ref{fig:distance from QFIM} we can achieve a very good approximation of the quantum Fisher information with a small number of repetitions (usually much less than $m$ repetitions). 

\begin{defi}
    \emph{(Unitary $k$-design) Let $\nu$ be a probability distribution defined over a set of unitaries $S\subseteq U(d)$. The distribution $\nu$ is unitary $k$-design if and only if:}
    \begin{equation}
       \mathbb{E}_{V\sim \nu}[V^{\otimes k}O V^{\dagger \otimes k}] =  \mathbb{E}_{U\sim \mu_H} [U^{\otimes k}O U^{\dagger \otimes k}]   
    \end{equation}
    \emph{for all $O\in \mathcal{L}((\mathbb{C}^d)^{\otimes k})$}.
\end{defi}

In general, generating Haar random unitaries on a quantum computer is a computationally exhaustive task, since most unitary operators require a number of gates that scale exponentially to the number of qubits \cite{mele2024introduction}. On the other hand, $k$-designs are distributions that match the Haar moments up to the $k$-th order (see Definition \ref{defi:k-design}). The advantage is that $k$-designs can be generated efficiently. As a result, we provide the following Corollary.

\begin{corollary}
    
    If $U$ is sampled from a 2-design $\nu$, then the matrix elements of the quantum Fisher information matrix can be calculated as:
      \begin{equation}
        [\mathcal{F}_Q(\boldsymbol{\theta})]_{ij} = 2(2^n + 1)    \sum_{\boldsymbol{s}} \mathbb{E}_{U\sim \nu}\Bigg[\frac{\partial p_{\boldsymbol{s}}^U(\boldsymbol{\theta})}{\partial \theta_i}\frac{\partial p_{\boldsymbol{s}}^U(\boldsymbol{\theta})}{\partial \theta_j}\Bigg]
    \end{equation}
\end{corollary}
\begin{proof}
    We can express the quantity:
    \begin{equation*}
        \mathbb{E}_{U\sim \mu_H}\Bigg[\frac{\partial p_{\boldsymbol{s}}^U(\boldsymbol{\theta})}{\partial \theta_i}\frac{\partial p_{\boldsymbol{s}}^U(\boldsymbol{\theta})}{\partial \theta_j}\Bigg]
    \end{equation*}
    as:
    \begin{equation}
    \begin{gathered}
        \mathbb{E}_{U\sim \mu_H}\Bigg[\frac{\partial p_{\boldsymbol{s}}^U(\boldsymbol{\theta})}{\partial \theta_i}\frac{\partial p_{\boldsymbol{s}}^U(\boldsymbol{\theta})}{\partial \theta_j}\Bigg] = \mathbb{E}_{U\sim \mu_H}\Bigg[\Tr[U \Pi_{\boldsymbol{s}} U^\dagger \frac{\partial \rho(\boldsymbol{\theta})}{\partial \theta_i}] \Tr[U \Pi_{\boldsymbol{s}} U^\dagger \frac{\partial \rho(\boldsymbol{\theta})}{\partial \theta_j}] \Bigg] \\=   \Tr \Big[\mathbb{E}_{U\sim \mu_H}[U^{\otimes 2} \Pi_{\boldsymbol{s}}^{\otimes 2} U^{\dagger \otimes 2}] \frac{\partial \rho(\boldsymbol{\theta})}{\partial \theta_i} \otimes \frac{\partial \rho(\boldsymbol{\theta})}{\partial \theta_j}\Big]
    \end{gathered}
    \end{equation}
    Thus, using Definition \ref{defi:k-design} for the unitary $k$-designs (for $O=\Pi_s ^{\otimes 2}$) we can conclude that if $U$ comes from a 2-design then the proposition holds.
\end{proof}

As a result, the quantum Fisher information can be estimated by sampling unitaries that come from a $k$-design with $k\geq 2$ (since any $k$-design is also a 2-design if $k\geq 2$). An example of such ensembles is the $n$-qubit Clifford group $Cl(n)$ which forms a 3-design. The Clifford group is defined as:
\begin{equation}
    Cl(n) := \{U \in U(2^n)| \; \: UPU^\dagger \in \mathcal{P}_n \; \text{for all } P\in \mathcal{P}_n \}
\end{equation}
where $\mathcal{P}_n$ is the Pauli group. Elements from the $n$-qubit Clifford group can be generated by a circuit with at most $\mathcal{O}(n^2/\log n)$ elementary gates \cite{aaronson2004improved}.\\

At the same time, the classical Fisher information matrix (when the parameterized quantum state is rotated by a global random unitary $U$ and then measured in the computational basis) can be expressed as:
\begin{equation}
    [\mathcal{F}_C^U(\boldsymbol{\theta})]_{ij} = -\sum_{\boldsymbol{s}} p_{\boldsymbol{s}}^U(\boldsymbol{\theta})\frac{\partial^2\ln[p_{\boldsymbol{s}}^U(\boldsymbol{\theta})]}{\partial \theta_i \partial \theta_j} = \sum_{\boldsymbol{s}} \frac{1}{p_{\boldsymbol{s}}^U(\boldsymbol{\theta})}\frac{\partial p_{\boldsymbol{s}}^U(\boldsymbol{\theta})}{\partial \theta_i}\frac{\partial p_{\boldsymbol{s}}^U(\boldsymbol{\theta})}{\partial \theta_j} 
\end{equation}
since: 
\begin{equation}
    \begin{gathered}
        -\frac{\partial^2}{\partial \theta_i \partial \theta_j} \ln p_{\boldsymbol{s}}^U(\boldsymbol{\theta}) = -\frac{\partial}{\partial \theta_i} \Bigg[\frac{1}{p_{\boldsymbol{s}}^U(\boldsymbol{\theta})} \frac{p_{\boldsymbol{s}}^U(\boldsymbol{\theta})}{\partial \theta_j}\Bigg] =
        -\frac{1}{p_{\boldsymbol{s}}^U(\boldsymbol{\theta})}\frac{\partial^2 p_{\boldsymbol{s}}^U(\boldsymbol{\theta})}{\partial \theta_i \partial \theta_j} + \frac{1}{(p_{\boldsymbol{s}}^U(\boldsymbol{\theta})^2)}\frac{\partial p_{\boldsymbol{s}}^U(\boldsymbol{\theta})}{\partial \theta_i}\frac{\partial p_{\boldsymbol{s}}^U(\boldsymbol{\theta})}{\partial \theta_j}
    \end{gathered}
\end{equation}
and we used again Eq. \eqref{eq:partial_zero}.\\

\noindent \textbf{Conjecture:} \emph{The average classical Fisher information matrix, when the parameterized quantum state $\ket{\phi(\boldsymbol{\theta})}$ is rotated by a random unitary $U$ and then measured in the computational basis approximates the quantum Fisher information matrix as:}
\begin{equation}
    \mathbb{E}_{U \sim \mu_H}[\mathcal{F}_C^U(\boldsymbol{\theta})] = \frac{1}{2} \mathcal{F}_Q(\boldsymbol{\theta})
\end{equation}  

One would have to show that:
\begin{equation}
    \sum_{\boldsymbol{s}} \mathbb{E}_{U\sim \mu_H}\Bigg[ p_{\boldsymbol{s}}^U(\boldsymbol{\theta})\frac{\partial^2\ln[p_{\boldsymbol{s}}^U(\boldsymbol{\theta}+ \boldsymbol{\epsilon})]}{\partial \epsilon_i \partial \epsilon_j} \Bigg{|}_{\boldsymbol{\epsilon}=0}\Bigg] = -(2^n + 1) \sum_{\boldsymbol{s}} \mathbb{E}_{U\sim \mu_H}\Bigg[p_{\boldsymbol{s}}^U(\boldsymbol{\theta})  \frac{\partial^2 p_{\boldsymbol{s}}^U(\boldsymbol{\theta}+ \boldsymbol{\epsilon})]}{\partial \epsilon_i \partial \epsilon_j} \Bigg{|}_{\boldsymbol{\epsilon}=0}\Bigg]
\end{equation}
where $p_s^U(\boldsymbol{\theta}) = \Tr [\rho(\boldsymbol{\theta}) U^\dagger \Pi_{\boldsymbol{s}} U]$.
Proving the above conjecture is a very challenging task. The main reason is that it requires the calculation of a Haar integral over the unitary group $U(2^n)$ where the unitaries rise in a non-linear way. Equivalently, it cannot be written as a $k$-moment of an operator for which ways to calculate the integrals are known (e.g. see \eqref{eq:haar_moments}). The proof of this conjecture is left for future work.

\section{Proof of Lemma \ref{lemma:descent_direction}}
\label{appendix:proof_of_lemma1}

    Consider the expectation value of the Hamiltonian $H$ of a parameterized quantum state $\ket{\phi(\boldsymbol{\theta})}$:
    \begin{equation*}
        E_{\tau}(\boldsymbol{\theta}) = \bra{\phi[\boldsymbol{\theta}(\tau)]}H\ket{\phi[\boldsymbol{\theta}(\tau)]}
    \end{equation*}
Its time derivative is then:
\begin{gather*}
    \frac{d}{d\tau}E_{\tau}(\boldsymbol{\theta}) = 2\Re \left(\bra{\phi[\boldsymbol{\theta}(\tau)]}H\frac{d\ket{\phi[\boldsymbol{\theta}(\tau)]}}{d\tau}\right) \\
    = 2\Re \left(\bra{\phi[\boldsymbol{\theta}(\tau)]}H\sum_{j=1}^m\frac{\partial\ket{\phi[\boldsymbol{\theta}(\tau)]}}{\partial\theta_j}\dot{\theta}_j\right)\\
    = \sum_{j=1}^m 2\Re \left(\bra{\phi[\boldsymbol{\theta}(\tau)]}H\frac{\partial\ket{\phi[\boldsymbol{\theta}(\tau)]}}{\partial\theta_j}\dot{\theta}_j\right) \\
    =  (\grad_{\boldsymbol{\theta}} E_\tau(\boldsymbol{\theta}))^\intercal\dot{\boldsymbol{\theta}} \\=- (\grad_{\boldsymbol{\theta}} E_\tau(\boldsymbol{\theta}))^\intercal [\mathbb{E}_{U \sim \nu}[\mathcal{F}_C^U(\boldsymbol{\theta}(\tau)]]^{-1} \grad_{\boldsymbol{\theta}} E_\tau(\boldsymbol{\theta})
    \end{gather*}
Since any classical Fisher information matrix is a positive semi-definite matrix, their average will also be positive semidefinite:
\begin{equation}
    \mathbb{E}_{U \sim \nu}[\mathcal{F}_C^U(\boldsymbol{\theta}(\tau)] \succcurlyeq 0
\end{equation}
which implies that its inverse is also positive semidefinite. As a result,
\begin{equation}
     \frac{d}{d\tau}E_{\tau}(\boldsymbol{\theta}) \leq 0
\end{equation}
and so we move into a \emph{a descent direction}. 

\section{Proof of Lemma \ref{lemma:error_in_approximation}}
\label{appendix:error_in_approximation}

Let the quantum Fisher information matrix $\mathcal{F}_Q$ and its estimator $\tilde{\mathcal{F}}_Q$ be non-singular with their eigenvalues satisfying:
    \begin{gather*}
        \lambda_1(\mathcal{F}_Q) \geq \lambda_2(\mathcal{F}_Q) \geq \ldots \geq \lambda_m(\mathcal{F}_Q) > 0 \\
        \lambda_1(\tilde{\mathcal{F}}_Q) \geq \lambda_2(\tilde{\mathcal{F}}_Q) \geq \ldots \geq \lambda_m(\tilde{\mathcal{F}}_Q) > 0
    \end{gather*}
    Consider the two different linear systems in Eqs. \eqref{eq:diff_eq_imaginary}, \eqref{eq:rmite_update} with their corresponding solutions $\dot{\boldsymbol{\theta}}_Q$ and $\dot{\tilde{\boldsymbol{\theta}}}_Q$ respectively. We have that:
    \begin{equation*}
    \begin{gathered}
        \norm{\dot{\boldsymbol{\theta}}_Q - \dot{\tilde{\boldsymbol{\theta}}}_Q} = \norm{\big(\tilde{\mathcal{F}}_Q^{-1} - \mathcal{F}_Q^{-1}\big)\grad_{\boldsymbol{\theta}} E_\tau} \\ \leq \norm{\tilde{\mathcal{F}}_Q^{-1} - \mathcal{F}_Q^{-1}} \norm{\grad_{\boldsymbol{\theta}} E_\tau} \leq \norm{\tilde{\mathcal{F}}_Q^{-1} - \mathcal{F}_Q^{-1}}\norm{\mathcal{F}_Q}\norm{\dot{\boldsymbol{\theta}}_Q}
    \end{gathered}
    \end{equation*}
where we used the fact that $\norm{\grad_{\boldsymbol{\theta}} E_\tau} \leq \norm{\mathcal{F}_Q}\norm{\dot{\boldsymbol{\theta}}_Q}$. As such, the relative error is upper bounded as:
\begin{equation}
    \frac{\norm{\dot{\boldsymbol{\theta}}_Q- \dot{\tilde{\boldsymbol{\theta}}}_Q}}{\norm{\dot{\boldsymbol{\theta}_Q}}} \leq \norm{\tilde{\mathcal{F}}_Q^{-1} - \mathcal{F}_Q^{-1}}\norm{\mathcal{F}_Q}
\label{eq:relative_error_bound1}
\end{equation}


 If we consider the case where estimator differs from the quantum Fisher information matrix by a small matrix $\Delta$ (with $\norm{\Delta}\leq \epsilon$), then we have:
\begin{gather*}
    \mathcal{F}_Q = \tilde{\mathcal{F}}_Q + \Delta \implies \\
    \tilde{\mathcal{F}}_Q^{-1} = \mathcal{F}_Q^{-1} + \tilde{\mathcal{F}}_Q^{-1}(\mathcal{F}_Q - \tilde{\mathcal{F}}_Q)\mathcal{F}_Q^{-1}
\end{gather*}
where:
\begin{gather*}
    \norm{\tilde{\mathcal{F}}_Q^{-1}(\mathcal{F}_Q - \tilde{\mathcal{F}}_Q)\tilde{\mathcal{F}}_Q^{-1}} \leq \norm{\tilde{\mathcal{F}}_Q^{-1}}\norm{\mathcal{F}_Q - \tilde{\mathcal{F}}_Q}\norm{\mathcal{F}_Q^{-1}}\\
    \leq \frac{\epsilon}{\lambda_{\min}(\tilde{\mathcal{F}}_Q) \lambda_{\min}(\tilde{\mathcal{F}}_Q)} = \frac{\epsilon}{\lambda_m(\tilde{\mathcal{F}}_Q)\lambda_m(\mathcal{F}_Q)}
\end{gather*}
If we assume that the matrix $\tilde{\mathcal{F}}_Q^{-1}(\mathcal{F}_Q - \tilde{\mathcal{F}}_Q)\mathcal{F}_Q^{-1}$ is also small then we can use the dual Weyl's inequality and have that:
\begin{equation}
    \begin{gathered}
        \lambda_1(\tilde{\mathcal{F}}_Q^{-1}) \geq \lambda_m(\mathcal{F}_Q^{-1}) + \lambda_1(\tilde{\mathcal{F}}_Q^{-1} -\mathcal{F}_Q^{-1}) \implies \\
        \lambda_1(\tilde{\mathcal{F}}_Q^{-1} -\mathcal{F}_Q^{-1}) \leq \lambda_1(\tilde{\mathcal{F}}_Q^{-1})  - \lambda_m(\mathcal{F}_Q^{-1}) \implies \\
        \lambda_1(\tilde{\mathcal{F}}_Q^{-1}-\mathcal{F}_Q^{-1}) \leq \frac{1}{\lambda_m(\tilde{\mathcal{F}}_Q)} - \frac{1}{\lambda_1(\mathcal{F}_Q)}
    \end{gathered}
\end{equation}
In that case, putting everything back in Eq. \eqref{eq:relative_error_bound1}, the relative error can be upper bounded as:
\begin{equation}
    \frac{\norm{\dot{\boldsymbol{\theta}}_Q - \dot{\tilde{\boldsymbol{\theta}}}_Q}}{\norm{\dot{\boldsymbol{\theta}}_Q}} \leq \frac{\lambda_1(\mathcal{F}_Q)}{\lambda_m(\tilde{\mathcal{F}}_Q)} - 1
\end{equation}

\section{Quantum Resources of Estimators}
\label{appendix:quantum_resources}

It is valuable to understand how the proposed estimators in Eqs. \eqref{eq:quantum_from_random_measurements_2_design} and \eqref{eq:cfisher_haar} scale with the number of parameters. It is true that the number of quantum resources needed to calculate the quantum Fisher information matrix scale quadratically with the number of parameters. As such, our estimators can be considered useful only if we can get a good approximation of the QFIM with quantum resources less than that required to calculate the QFIM.

What can be considered a good approximation is user-dependent and can be inferred by the distance of the estimator from the exact matrix. Let $\tilde{\mathcal{F}}_Q$ be an estimator of the QFIM, and let $\mathcal{F}_Q$ be the exact QFIM. The error is then evaluated as:
\begin{equation}
    \norm{\tilde{\mathcal{F}}_Q - \mathcal{F}_Q}
\end{equation}
where we choose $\norm{\cdot}$ to be the $l_2$ norm. In this paper, we tested two values for the desired threshold. That is, we set target error to $\epsilon = 0.1$ and $\epsilon=0.2$ and aimed to quantify how much quantum resources are needed as the number of parameters increases, in order to achieve an error $\norm{\tilde{\mathcal{F}}_Q - \mathcal{F}_Q} \leq \epsilon$. To do this, we employed the ansatz family depicted in Figure \ref{fig:hardware_efficient_ansatz}. If $l$ is the number of layers of the ansatz family, its parameters increase as $m =2(l+1)n$, where $n$ is the number of qubits. For our experiments, we employed both 10, 11 and 12 qubit instances with $l\in [3,4,5,6]$. For each instance, we prepared a random state $\ket{\psi(\boldsymbol{\theta})}$ and used the estimators \eqref{eq:quantum_from_random_measurements_2_design} and \eqref{eq:cfisher_haar} to estimate the number of samples needed to achieve error less than $\epsilon$. For each layer and qubit choice, we sampled 10 random instances and then calculated the average number of quantum states needed to achieve the desired error. Our results are illustrated in Figure \ref{fig:quantum_resources_different_layers}.

As it is clearly visualized, the average classical Fisher information estimator in Eq. \eqref{eq:cfisher_haar} is able to approximate with much fewer resources the QFIM. For every choice of error tolerance, number of qubits and number of parameters, the estimator requires fewer quantum resources than the exact QFIM. However, the same is not true for the 2-design estimator. The latter achieves a good approximation to the QFIM slower than the average Fisher estimator. As is illustrated, the 2-design estimator becomes useful in the large parameter regime, which is ideal since the calculation of the QFIM for a small number of parameters is not computationally expensive. We also observe that if we allow for a slightly larger error (i.e. $\epsilon = 0.2$), then both estimators achieve a very fast convergence. However, it is important to also mention that there is a critical error threshold $\epsilon_c$ after which, calculating the full QFIM requires less quantum states than any of the 2 estimators.

\begin{figure*}
\begin{tikzpicture}
\node (img1)  {\includegraphics[scale=0.35]{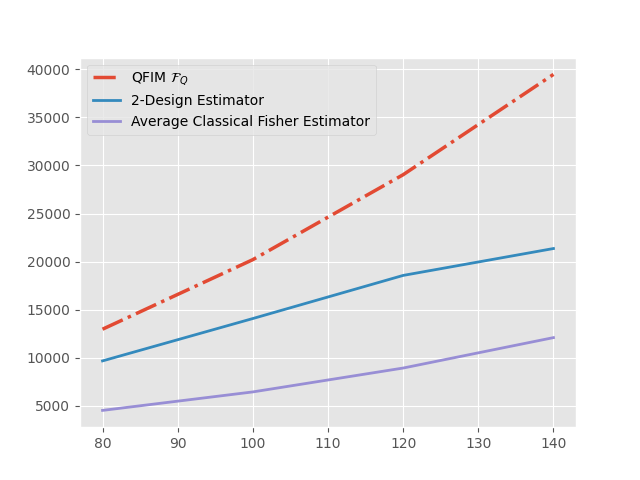}};
\node[below=of img1, node distance=0cm, yshift=1.1cm] {\scriptsize Number of Parameters};
\node[left=of img1, node distance=0cm, rotate=90, anchor=center,yshift=-1cm] {\scriptsize Quantum State Preparations};
\node[right=of img1, xshift=-1cm] (img2)  {\includegraphics[scale=0.35]{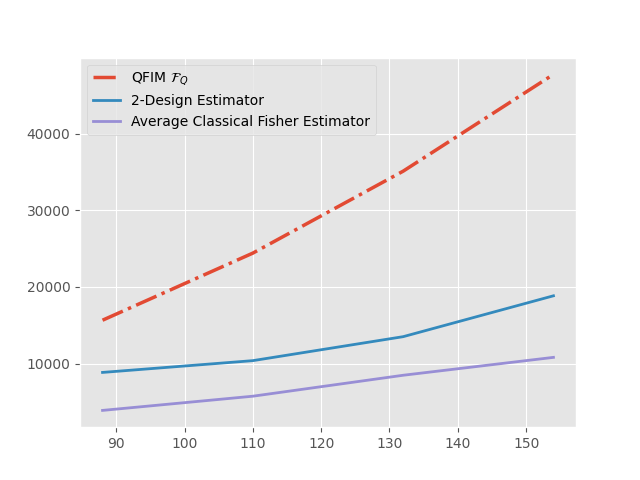}};
\node[above=of img1, node distance=0cm, yshift=-1.5cm] {\scriptsize $\epsilon=0.2$, 10 Qubits};
\node[below=of img2, node distance=0cm, yshift=1.1cm] {\scriptsize Number of Parameters};
\node[left=of img2, node distance=0cm, rotate=90, anchor=center,yshift=-1cm] {\scriptsize Quantum State Preparations};
\node[above=of img2, node distance=0cm, yshift=-1.5cm] {\scriptsize $\epsilon=0.2$, 11 Qubits};
\node[below=of img1] (img3) {\includegraphics[scale=0.35]{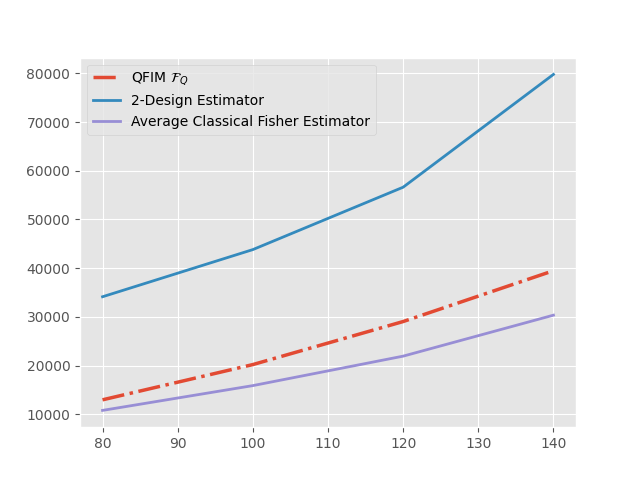}};
\node[below=of img3, node distance=0cm, yshift=1.1cm]{\scriptsize Number of Parameters};
\node[left=of img3, node distance=0cm, rotate=90, anchor=center, yshift=-1cm] {\scriptsize Quantum State Preparations};
\node[right=of img3, xshift=-1cm] (img4)  {\includegraphics[scale=0.35]{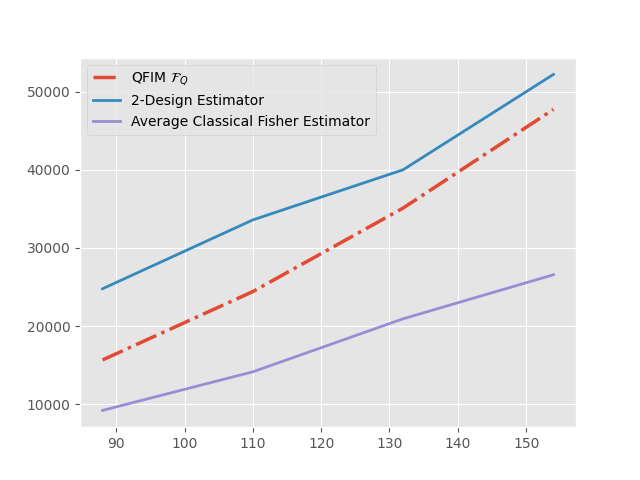}};
\node[above=of img3, node distance=0cm, yshift=-1.5cm] {\scriptsize $\epsilon=0.1$, 10 Qubits};
\node[below=of img4, node distance=0cm, yshift=1.1cm] {\scriptsize Number of Parameters};
\node[left=of img4, node distance=0cm, rotate=90, anchor=center,yshift=-1cm] {\scriptsize Quantum State Preparations};
\node[right=of img2, xshift=-1cm] (img5)  {\includegraphics[scale=0.35]{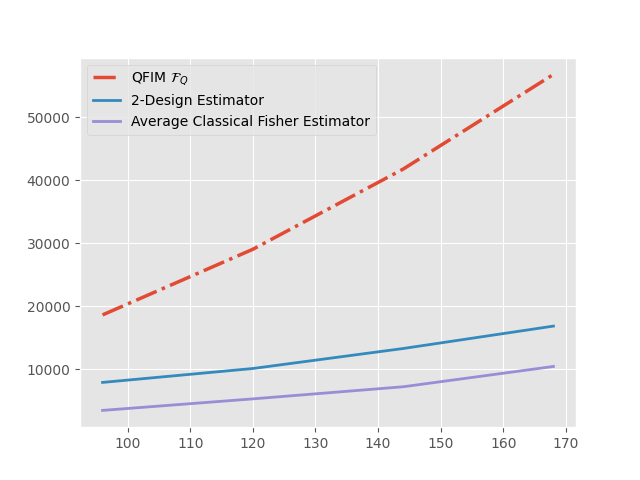}};
\node[above=of img4, node distance=0cm, yshift=-1.5cm] {\scriptsize $\epsilon=0.1$, 11 Qubits};
\node[below=of img5, node distance=0cm, yshift=1.1cm] {\scriptsize Number of Parameters};
\node[left=of img5, node distance=0cm, rotate=90, anchor=center,yshift=-1cm] {\scriptsize Quantum State Preparations};
\node[above=of img5, node distance=0cm, yshift=-1.5cm] {\scriptsize $\epsilon=0.2$, 12 Qubits};
\node[right=of img4, xshift=-1cm] (img6)  {\includegraphics[scale=0.35]{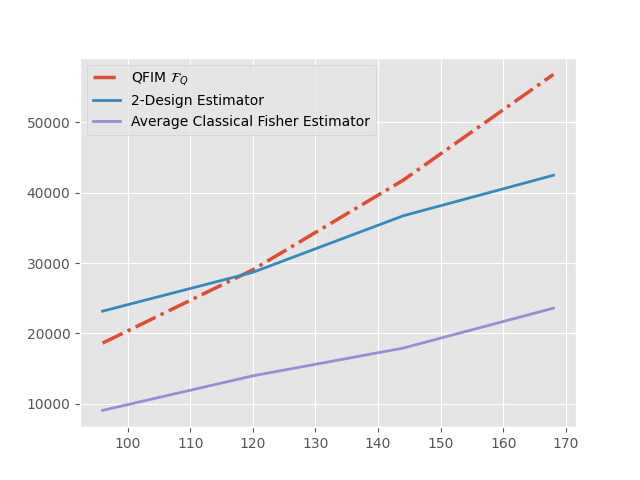}};
\node[below=of img6, node distance=0cm, yshift=1.1cm] {\scriptsize Number of Parameters};
\node[left=of img6, node distance=0cm, rotate=90, anchor=center,yshift=-1cm] {\scriptsize Quantum State Preparations};
\node[above=of img6, node distance=0cm, yshift=-1.5cm] {\scriptsize $\epsilon=0.1$, 12 Qubits};
\end{tikzpicture}
\caption{Comparison of quantum resources (different quantum state preparations) for the 2 estimators in Eqs \eqref{eq:quantum_from_random_measurements_2_design} and \eqref{eq:cfisher_haar}, for different number of parameters and target error approximations. The average classical Fisher estimator is able to outperform the 2-design estimator, approximating the QFIM with high accuracy and requiring significantly fewer quantum calls.}
\label{fig:quantum_resources_different_layers}
\end{figure*}

\end{document}